\documentclass[10pt,conference,compsocconf,romanappendices]{IEEEtran}
\pdfoutput=1
\usepackage{tikz}
\usetikzlibrary{shapes,arrows}
\usepackage[skip=2pt]{caption}
\usepackage{subcaption}
\usepackage{float}
\usepackage{epsfig}
\usepackage{xspace}
\usepackage{color}
\usepackage{pgfplots}
\usetikzlibrary{matrix}
\usetikzlibrary{patterns}
\usepgfplotslibrary{groupplots}
\pgfplotsset{compat=newest}

\usepackage[hyphens]{url}
\usepackage{algorithm}
\usepackage{algpseudocode}
\usepackage{amsthm}  
\usepackage{framed}
\usepackage{thm-restate}
\usepackage{ifthen}

\newboolean{short}   
\setboolean{short}{false}
\usepackage[hidelinks]{hyperref} 

\usepackage{soul} \setul{0.5ex}{0.3ex} 
\setulcolor{red}

\newcommand{\noop}[1]{}
\newcommand{\tool}{{\sc OverlayBB}\xspace}
\newcommand{\codename}{{\sc BCube}\xspace}

\newcommand{\Paragraph}[1]{\smallskip\noindent{\bf #1}}
\newtheorem{theorem}{Theorem}

\newtheorem{lemma}[theorem]{Lemma}

\renewcommand{\paragraph}[1]{\pagebreak\mbox{XXXXX} \pagebreak\mbox{XXXXX} \pagebreak\mbox{XXXXXX}}

\usepackage{cite}
\usepackage{amsmath,amssymb,amsfonts}
\usepackage{graphicx}
\usepackage{textcomp}
\usepackage{xcolor}
\def\BibTeX{{\rm B\kern-.05em{\sc i\kern-.025em b}\kern-.08em
    T\kern-.1667em\lower.7ex\hbox{E}\kern-.125emX}}
\usepackage{lastpage}
\usepackage{fancyhdr}
 
\usepackage{lipsum}
\usepackage[cp1252]{inputenc} \graphicspath{images/} \DeclareMathAlphabet{\mathpzc}{OT1}{pzc}{m}{it}

\begin{document}
\title{Using Throughput-Centric Byzantine Broadcast \\
	to Tolerate Malicious Majority in Blockchains}

\author{\IEEEauthorblockN{Ruomu Hou}
	\IEEEauthorblockA{National University of Singapore\\
		houruomu@comp.nus.edu.sg}
	\and
	\IEEEauthorblockN{Haifeng Yu}
	\IEEEauthorblockA{National University of Singapore\\
		haifeng@comp.nus.edu.sg}
	\and
	\IEEEauthorblockN{Prateek Saxena}
	\IEEEauthorblockA{National University of Singapore\\
		prateeks@comp.nus.edu.sg}
	\and}

\maketitle

\begin{abstract}
	
Fault tolerance of a blockchain is often characterized by the fraction $f$ of ``adversarial power'' that it can tolerate in the system. Despite the fast progress in blockchain designs in recent years, existing blockchain systems can still only tolerate $f$ below $0.5$. Can practically usable blockchains tolerate a malicious majority, i.e., $f$ above $0.5$?

This work presents a positive answer to this question. We first note that the well-known impossibility of {\em byzantine consensus} for $f$ above $0.5$ does not carry over to blockchains. To tolerate $f$ above $0.5$, we use {\em byzantine broadcast}, instead of byzantine consensus, as the core of the blockchain. A major obstacle in doing so, however, is that the resulting blockchain may have extremely low throughput. To overcome this central technical challenge, we propose a novel byzantine broadcast protocol \tool, that can tolerate $f$ above $0.5$ while achieving good throughput. Using \tool as the core, we present the design, implementation, and evaluation of a novel Proof-of-Stake blockchain called \codename. \codename can tolerate a malicious majority, while achieving practically usable transaction throughput and confirmation latency in our experiments with $10000$ nodes and under $f = 0.7$. To our knowledge, \codename is the first blockchain that can achieve such properties.
	
\end{abstract}

\section{Introduction}
\label{sec:intro}

Fault tolerance is a property of central importance in modern
distributed systems such as blockchains.
Fault tolerance is often characterized by the
fraction $f$ of ``adversarial power'' that a system can tolerate. Here the ``adversarial power'' may correspond to i) malicious nodes in permissioned blockchains, ii) adversarially-controlled computational power in Proof-of-Work-based permissionless blockchains, or iii) adversarially-controlled stake in Proof-of-Stake-based permissionless blockchains.
Bitcoin's consensus protocol, invented over a decade ago, can 
tolerate $f$ below $\frac{1}{2}$.
While subsequent blockchain systems (e.g., \cite{solida,prism19,ohie20,snow-white,algorand,ouroboros,omniledger,elastico,rapidchain}) have achieved significantly better performance than
Bitcoin, all of them still can only tolerate $f$ below
$\frac{1}{2}$, or sometimes even lower. This is regardless of whether these designs are for permissioned systems, or for permissionless systems using Proof-of-Work/Proof-of-Stake.

There are growing desires, however, for blockchains to tolerate $f\ge \frac{1}{2}$. For example, there have been
double-spending attacks on public blockchains, where malicious actors
temporarily control more than half of the adversarial power in the
network~\cite{bitcoin-gold,krypton,shift,bitcoin-cash}. Similarly, it
has been highlighted that a few centralized miners control more than
half of the power in many Proof-of-Work and Proof-of-Stake
blockchains~\cite{kwon-aft19}. 

One reason why blockchains today cannot tolerate $f\ge \frac{1}{2}$ is that they often build upon {\em byzantine consensus}~\cite{lynch1996distributed}. Byzantine consensus is a {\em one-shot} game --- among other things, it requires that if honest nodes all have the same proposal for the next block, then they must all decide on that block, instead of on some adversarially-chosen block. Such requirement\footnote{A similar requirement has led to the recent impossibility result in \cite{garay20}.} makes it impossible to tolerate $f\ge \frac{1}{2}$.
A blockchain, on the other hand, is a continuous process where usually each block in the distributed ledger is proposed by a random proposer. It is acceptable for the block to be adversarially-chosen, if the proposer happens to be malicious. 
Hence impossibility results on byzantine consensus under $f\ge \frac{1}{2}$ do not necessarily carry over to blockchains.

\Paragraph{Byzantine broadcast.}
In our pursuit of blockchains that can tolerate $f \ge \frac{1}{2}$, 
we have revisited various classical primitives. We eventually focus on one such primitive --- 
{\em byzantine
  broadcast}~\cite{chan20,dolev83,ganesh16,hirt14,nayak20,tsimos20,wan20}.  In byzantine broadcast, there is a single publicly known
{\em broadcaster} that broadcasts an {\em object} (e.g., a block in blockchain) to
all nodes. Some of the nodes, including the broadcaster, may be
malicious. A byzantine broadcast protocol guarantees:
\begin{itemize}
	\item All honest (i.e., non-malicious) nodes eventually output
          the same object (i.e., {\em agreement}). This object is
          allowed to be a special null object. 

        \item If the broadcaster is honest, all honest
          nodes must output the object broadcast by the broadcaster.
\end{itemize}
Byzantine
broadcast is closely related to byzantine consensus, but crucially
differs from it --- in particular, byzantine broadcast is solvable in synchronous systems for all $f < 1$. 
Starting from now on, {\em this paper will only be concerned with byzantine broadcast protocols that can tolerate  $f\ge \frac{1}{2}$}. 

While rarely mentioned in literature, byzantine broadcast can be used~\cite{pass2018thunderella} to build blockchains. In particular, if we use a byzantine broadcast protocol that can tolerate $f \ge \frac{1}{2}$, then the resulting blockchain will immediately be able to tolerate $f\ge \frac{1}{2}$ as well. But at the same time, the resulting blockchain's throughput will also inherit from the throughput of the byzantine broadcast protocol. This turns out to be the major obstacle, because existing byzantine broadcast protocols~\cite{chan20,dolev83,ganesh16,hirt14,nayak20,tsimos20,wan20} seriously fall short of providing acceptable throughput, as shown next.

\Paragraph{Throughput of existing byzantine broadcast protocols.}
Let us clearly define throughput. Imagine that each node has $\mathbb{B}$ available
bandwidth, as provided by the deployment environment. 
Under such a constraint, let $x$ be the total
number of bits that a byzantine broadcast protocol can broadcast
within a time period $y$. We define the protocol's {\em throughput} to be
$\mathbb{T} = x/y$, and
define the {\em normalized
	throughput} to be $\mathbb{R} = \mathbb{T}/\mathbb{B}$. We also
call $\mathbb{R}$ as the {\em throughput-to-bandwidth ratio} or {\em
	TTB ratio}. This ratio essentially serves to isolate the inherent merit of the protocol from the goodness of the deployment environment, since the throughput achievable by a protocol naturally increases when deployed in a better environment offering higher $\mathbb{B}$. Obviously, $\mathbb{R}$ is always between 0 and 1, and larger is better. 

Existing byzantine broadcast protocols unfortunately have rather low TTB ratios. 
This is perhaps not surprising, since throughput or TTB ratio has not been explicitly considered in prior research on byzantine broadcast.\footnote{Prior works have considered {\em communication complexity} (CC), which is the total number of bits sent by all honest nodes. But CC does not map to throughput or TTB ratio. For example, some protocols~\cite{chan20,dolev83,nayak20} require a node to send many bits in a few ``busy'' rounds, and nothing in other rounds. (For a given node, the adversary decides which rounds are ``busy'' for that node.) While such protocols may have low CC, the deployment environment still needs to provision for the high bandwidth need of those ``busy'' rounds.}
For example, the Dolev-Strong protocol~\cite{dolev83} and its variant~\cite{tsimos20} have $\mathbb{R}< \frac{1}{wfn}$ (see analysis in Section~\ref{sec:background2}), where $n$ is the total number of nodes and $w$ is the degree of a node in the overlay network. For $n = 10000$ and $w=40$, the protocol has  $\mathbb{R}< 3.6\times 10^{-6}$ when $f = 0.7$. If every node has $20$Mbps available bandwidth, then the protocol's throughput will be less than $0.072$Kbps, even if we ignore all other overheads in implementation. A blockchain built upon such a protocol 
will also have a throughput of less than $0.072$Kbps, which is practically unusable.
As another example, in our experiments, the state-of-the-art design recently proposed by Chan et al.~\cite{chan20} achieves a throughput
of only about $0.45$Kbps under $20$Mbps available bandwidth.
See Section~\ref{sec:background2} and \ref{sec:related} for more discussions on existing protocols.

\Paragraph{Our \tool protocol.}
As the first contribution of this paper, we propose a novel byzantine
broadcast protocol called \tool. \tool is particularly suitable for large-scale
systems where nodes communicate via a multi-hop overlay. 
\tool achieves $\mathbb{R}= \Theta(\frac{1}{w})$, by using fragmentation, proper delay/compensation in fragment propagation, and other techniques. Here $w$ is the degree of the
nodes in the overlay.
For example, if the overlay is a random graph,
then $w$ can be just $O(\log n)$. This $\Theta(\frac{1}{w})$ TTB ratio is significantly better than existing protocols. 
(We show later that \tool achieves a throughput of roughly $163$Kbps under $20$Mbps available bandwidth.)

\Paragraph{From byzantine broadcast to blockchain.}
As the second contribution of this paper, we present the design,
implementation, and evaluation of a novel Proof-of-Stake blockchain called 
\codename (i.e., Byzantine Broadcast-based Blockchain, or B$^3$). 
Using \tool as its core,
\codename can tolerate $f \ge \frac{1}{2}$, while achieving {\em practically usable}
transaction throughput and confirmation latency. Specifically,
we have implemented a prototype of \codename, and evaluated its performance, with up to $10000$ nodes and under similar configurations as in prior works~\cite{algorand,ohie20}. 
In our experiments with $f=0.7$ and a target error probability of $\epsilon \le 2^{-30}$, \codename achieves $163$Kbps throughput, with one $2$MB block generated about every 98 seconds, and has a transaction confirmation latency of less than $6$ hours.

Such performance of \codename is certainly not on par with blockchains that only tolerate $f < \frac{1}{2}$. 
But \codename's throughput and latency are nevertheless {\em practically usable}: As a reference point,
Bitcoin's throughput is about 14Kbps, with one 1MB block generated about every 600 seconds. Bitcoin 
entails a confirmation latency of about $9.3$ hours,\footnote{Based on the formula from \cite{kiffer18}, to ensure that the probability of the  
adversary (under all possible attack strategies) reverting a block is at most $2^{-30}$ in Bitcoin, the block needs to be at least $56$ blocks deep in the blockchain. Given Bitcoin's 10-minute inter-block time, this translates to $9.3$ hours. 
The well-known ``6 blocks deep'' rule of thumb in Bitcoin, and the corresponding 1-hour confirmation latency, would only give $\epsilon\approx 0.05$~\cite{kiffer18} under $f=0.25$.} based on the state-of-the-art analysis~\cite{kiffer18}, to achieve $\epsilon \le 2^{-30}$ under $f=0.25$. 

To our knowledge, \codename is the very 
first blockchain that can tolerate $f\ge \frac{1}{2}$, while achieving practically usable throughput and latency. There are only a few prior approaches~\cite{chan20,pass2018thunderella} for designing blockchains with $f \ge \frac{1}{2}$, which are all based on byzantine broadcast. The throughput achieved by those approaches (i.e., no more than $0.45$Kbps under same setting as \codename) is far from practically usable. Furthermore, all those prior works are purely theoretical, ignore various practical issues, and provide no implementation.

\Paragraph{Roadmap.}
The next section defines our system/attack model. Section~\ref{sec:background2} provides some background. Section~\ref{sec:design} and \ref{sec:pseudocode} describe the design of \tool. Section~\ref{sec:blockchain} presents the design of \codename. Section~\ref{sec:proof} gives the security analysis of \tool and \codename. Section~\ref{sec:exp} presents the implementation and evaluation of \codename.

\section{System Model and Attack Model}
\label{sec:model}

We model hash functions as random oracles. We assume that some initial trusted setup provides a {\em genesis block}, which contains an unbiased random beacon to be used in the very first epoch. This is a typical assumption in Proof-of-Stake blockchains (e.g., \cite{snow-white,algorand}).

\Paragraph{Nodes and stakes/coins.}
We consider a {\em permissionless} setting (i.e., similar to Algorand~\cite{algorand}), without PKI or initial trusted setup for binding nodes to identities. Each node in the system holds a locally-generated public-private key pair, and the public key is viewed as the node's {\em id}. Each {\em node} can either be {\em honest} or {\em malicious}. The malicious nodes are fully {\em byzantine}, and may deviate arbitrarily from the protocol. They may also arbitrarily collude, and we view them as all being controlled by the {\em adversary}. 
We allow the adversary to be {\em mildly-adaptive}~\cite{rapidchain, snow-white,ouroboros} --- for example, it takes multiple epochs for the adversary to adaptively corrupt a node. 

We rely on {\em Proof-of-Stake (PoS)}~\cite{snow-white,algorand,ouroboros} for Sybil defense in our permissionless setting: We assume that there are some {\em stakes} (or {\em coins}) in the system,
where the total number of coins may change over time.
At any point of time, each coin has an {\em owner}, which is the node holding that coin. 
Again, the owner may change over time. Information regarding which nodes hold which coins is stored in the blockchain itself, and is publicly known. 
We assume that at any point of time, at most $f$ fraction of the stakes/coins in the system are owned by malicious nodes, where $f$ is some constant no larger than $0.99$. (Our experiments mainly consider $f = 0.7$.) Sometimes as a stepping stone, we also consider a simplified permissioned setting with exactly $n$ nodes, where we use $f$ to denote the fraction of malicious nodes.

To simplify periodic beacon generation, \codename further relies on a weak Proof-of-Work (PoW) assumption: We assume that the adversary's computational power is at most $100$ times of the aggregate computational power of the honest nodes. This assumption is separate from and independent of the earlier $f$ threshold.
(If needed, this ``$100$'' value can be further increased without impacting security, but at the cost of lower performance.)
Note that our assumption differs from PoW-based blockchains, whose {\em security} depends on the adversary having {\em less computational power} than the honest nodes.

\Paragraph{Communication.}
We assume that all the honest nodes form a connected overlay network --- this is a typical assumption in large-scale blockchain systems (e.g., \cite{algorand}). 
Consider any two neighboring honest nodes $A$ and $B$ in the overlay. It will be convenient to view the undirected edge between $A$ and $B$ as two directed edges in two directions. With respect to some $\delta_1$ value, we say that the directed edge from $A$ to $B$ is {\em good} if a message sent by $A$ can reach (with proper retries) $B$ within $\delta_1$ time, as long as the message is relatively small (e.g, $\le 10$KB). Otherwise the edge is {\em bad}. In general, under reasonably large $\delta_1$  (e.g., $\delta_1 = 10$ seconds), one would expect that while some edges may occasionally be bad, most edges among the honest nodes will be good. Hence we assume the {\em honest subgraph} (i.e., the subgraph containing all the honest nodes and all the good edges) to be connected. We use $d$ to denote an upper bound on the diameter of this honest subgraph. 

Partitioning attacks~\cite{hijacking-bitcoin,tran20} can cause our assumption to be violated in general. But such attacks apply to many other existing blockchains as well (e.g., \cite{elastico, hijacking-bitcoin, marcus2018low, ohie20, ouroboros, tran20}),
{\em despite} that all these existing designs can only tolerate $f < \frac{1}{2}$. How to defend against such partitioning attacks is an active research topic by itself, and is beyond the scope of this paper: Possible defenses include hiding the overlay network structure~\cite{dandelion17}, diversifying neighbors' profile~\cite{heilman2015eclipse}, or preserving neighbors that provide fresher data~\cite{tran20}. 

We assume that nodes have loosely synchronized clocks, so that the clock readings on any two nodes do not differ by more than $\delta_2$ (e.g., $\delta_2 = 2$ second).
We will describe a byzantine broadcast execution as a sequence of {\em rounds}, and each node uses its local clock to keep track of the beginning of this execution as well as the current round number. We allow the starting time of each round on different nodes to be somewhat misaligned due to the $\delta_2$ clock error. Each round has a fixed duration $\delta = \delta_1+ \delta_2$ (e.g., $\delta = 12$ seconds). We assume that CPU processing delay is negligible, as compared to $\delta$. At the beginning of each round, a node receives messages, processes them, and then sends new messages. Since $\delta = \delta_1+ \delta_2$, a message sent in round $i$ along a good edge is received by the beginning of round $i+1$ on the receiver, as long as the message is relatively small (e.g, $\le 10$KB). 

\begin{table}[]
	\vspace*{-0mm}
\caption{Key notations.\label{tab:notation}}
\vspace*{-0mm}
	\begin{center}
		\begin{tabular}{|r|l|}
			\hline
			$n$&total number of nodes (for permissioned setting)\\ \hline
			$m$&number of nodes (for permissioned setting) in the committee, \\
			&or number of  coins (for PoS setting) held by committee members\\ \hline
			$f$&fraction of malicious nodes (for permissioned setting), \\
			&or fraction of coins (for PoS setting) held by malicious nodes\\ \hline
			$d$&upper bound on diameter of subgraph containing honest nodes/edges\\ \hline
			$w$&maximum number of neighbors (both honest neighbors \\
			&and malicious neighbors) that an honest node may have\\ \hline
			$\epsilon$&error probability \\ \hline
			$\delta$&round duration \\ \hline
			$l$&size of object to be broadcast in byzantine broadcast protocol\\ \hline
			$s$&total number of fragments (of the object to be broadcast) \\ \hline
		\end{tabular}
	\end{center}
\vspace*{-4mm}
\end{table}

\Paragraph{Problem definition.}
We aim to design a blockchain system where each node maintains an append-only sequence of blocks. (\codename has no forks, and all blocks in this sequence are considered as ``confirmed''.) Each block may contain, for example, a list of transactions. The blockchain should achieve standard {\em safety} and {\em liveness} guarantees, despite the byzantine behavior of all the malicious nodes. 
Roughly speaking, {\em safety} means that the sequences on all honest nodes are consistent with each other, while {\em liveness} means that the sequence on each honest node keeps growing over time. We defer the exact definitions to Section~\ref{sec:proof}.

\Paragraph{Notations.}
Table~\ref{tab:notation} summarizes the notations so far, and also defines several other notations.

\section{Background on Byzantine Broadcast}
\label{sec:background2}

We review two existing byzantine broadcast protocols~\cite{chan20,dolev83}, which \tool builds upon. To help understanding, we describe them in a simple permissioned setting with $n$ nodes, out of which $fn$ are malicious. We will first assume a clique topology among the $n$ nodes, and then generalize to arbitrary multi-hop topology.

\subsection{Dolev-Strong Protocol~\cite{dolev83}}
\label{sec:background3}

\Paragraph{Clique topology.}
In round $0$ of this protocol, the broadcaster sends the object to all nodes, with its own signature attached. Upon receiving an object in round $t$, if the object has less than $t$ signatures attached (including the broadcaster's signature), a node drops the object. Otherwise the node {\em accepts} this object, and then adds its own signature to the object and forwards the object to all other nodes. Once an object is accepted, a node will not forward the object again in the future. A node may accept more than one object, when the broadcaster is malicious. At the end of round $fn+1$, a node {\em outputs} the special null object $\bot$, if it has accepted more than one object (implying a {\em conflict}) or if it has accepted none. Otherwise it {\em outputs} the (single) object accepted.

The key intuition in this protocol is the following:
{\em When a node $B$ is about to forward/send an object, node $B$ can safely accept the object if $B$ knows that its send will cause all other honest nodes to accept this object (if they have not already done so).} 
This ensures that either all or none honest nodes accept that object.
Specifically in this protocol, if $B$ receives and then immediately forwards an object in round $t \le fn$, then $B$ is sure that all other honest nodes must receive and accept this object in round $t+1\le fn+1$, which is before the end of the protocol. On the other hand, if $B$ sends an object in round $fn+1$, then other nodes will not receive the object before the end of the execution. But in such a case, $B$ must have seen at least $fn+1$ signatures on the object. One of these must be from some honest node $A$, and $A$ must have previously already forwarded the object to all honest nodes. Namely, $A$ has already done the job for $B$.  

The protocol comes with a further optimization: Once a node has accepted two objects, it no longer accepts/forwards more objects. Hence a node only sends at most two messages throughout the execution. Agreement is still preserved: If one honest node $A$ accepts two objects, then another honest node $B$ must also accept two objects (which may be different from what $A$ accepts). Hence all honest nodes will output $\bot$.

\Paragraph{Multi-hop topology.}
The protocol naturally generalizes~\cite{dolev83,tsimos20} to multi-hop topologies.
The only modifications needed are: i) the protocol now runs for $fn+d$ rounds, and ii) a node now only forwards an object to its (up to $w$) neighbors.

\Paragraph{TTB ratio.}
Consider any honest node $A$. In the above protocol, there are some rounds (potentially chosen by the adversary) during which $A$ needs to forward objects to all its neighbors. Recall from Table~\ref{tab:notation} that $l$ is the object size. Hence in each of those rounds, $A$ needs to send at least $lw$ bits total.
Under the given bandwidth constraint $\mathbb{B}$, each node has the capacity to send at most $\mathbb{B}\delta$ bits in each round, where $\delta$ is the round duration. Hence the maximum $l$ the protocol can manage is $\frac{\mathbb{B}\delta}{w}$. The protocol has total $fn+d$ rounds. Since the protocol manages to broadcast an object of size $l = \frac{\mathbb{B}\delta}{w}$ using
total $(fn+d)\delta$ time, we have $\mathbb{T} = \frac{l}{(fn+d)\delta}$ and $\mathbb{R} = \mathbb{T}/\mathbb{B} = \frac{l}{(fn+d)\delta\mathbb{B}} \le \frac{\mathbb{B}\delta/w}{(fn+d)\delta\mathbb{B}} <\frac{1}{wfn}$. 

\subsection{Chan et al.'s Protocol~\cite{chan20}}
\label{sec:background4}

\Paragraph{Clique topology.}
Recently, Chan et al.~\cite{chan20}\footnote{Chan et al.'s original protocol~\cite{chan20} allows a fully-adaptive adversary, but can only broadcast messages containing a single-bit. The version we describe here is for mildly-adaptive adversaries and can broadcast multi-bit messages.} have proposed an elegant design to substantially reduce the number of rounds in the Dolev-Strong protocol.
Chan et al.'s protocol first selects a random committee of $m$ nodes. 
Now the $m$ committee members can do byzantine broadcast among themselves, using the Dolev-Strong protocol~\cite{dolev83} while taking at most $m$ rounds. But it is not immediately clear how the remaining non-committee members can decide. In particular, since a majority of the committee members can be malicious, voting will not work. 

Chan et al.~\cite{chan20} overcomes this problem in the following way. 
Consider one round in the Dolev-Strong protocol, where one committee member $A$ sends a message (containing the object and signatures) to all other committee members. Their idea~\cite{chan20} is to replace this round with two rounds, so that $A$ sends the message to all the non-committee members first, and then the non-committee members forward $A$'s message (unchanged) to all the committee members. (Hence there will be total $2m$ rounds.) This enables the non-committee members to observe the communication originated from the honest committee members. Before forwarding an object, a non-committee member $B$  can precisely predict (based on the signatures on the object) whether the committee members, upon receiving this object, will accept the object. If yes, $B$ accepts the object before forwarding it.

\Paragraph{Multi-hop topology.}
Chan et al.'s protocol trivially generalizes to a multi-hop topology, {\em assuming that each node  has sufficient bandwidth to relay all messages}. Specifically, whenever a node needs to send messages to other nodes, it simply does a {\em multicast} (i.e., flooding) on the multi-hop topology, taking $d$ rounds. Hence each of the $2m$ rounds in the clique setting now becomes $d$ rounds, and there are total $2dm$ rounds.

\section{Design of \tool}
\label{sec:design}

Byzantine broadcast protocols are often not complex in implementation, but their designs can be subtle. 
Because of this, this section focuses on intuitions. We do not aim to cover all possibilities, nor to rigorously argue for correctness here. Later, Section~\ref{sec:pseudocode} presents the complete pseudo-code of \tool, based on which Section~\ref{sec:proof} provides formal proof for correctness and analysis of the TTB ratio $\mathbb{R}$. Such an end-to-end proof is the only way to ultimately verify the protocol's correctness.

This section considers a multi-hop topology, but to help understanding, we still assume the simple permissioned setting with $n$ nodes. Section~\ref{sec:pseudocode} later generalizes to the PoS setting.

\subsection{Avoid Relaying Unlimited Number of Objects}
\label{sec:design1}

Chan et al.'s protocol~\cite{chan20} serves as a starting point of our design. 
When used in multi-hop topologies, their protocol relies on the implicit assumption that {\em a node has sufficient bandwidth to relay all possible multicast messages}. 
This section first shows that such an assumption can prevent the protocol from guaranteeing agreement in practice, namely, when $\mathbb{B}\ne \infty$.
We then propose a simple solution to fix this problem.
Section~\ref{sec:design2} and \ref{sec:design3} later propose more techniques to improve $\mathbb{R}$, 
to eventually get \tool.

\Paragraph{Chan et al.'s protocol in multi-hop overlay.}
In a multi-hop topology, each node in Chan et al.'s protocol simultaneously plays two roles: First, a node is either a committee member or a non-committee member. Second, a node is always a relaying node in the overlay for the purpose of multicast, and it needs to relay all multicast messages. 
Now a malicious broadcaster can generate many objects, all with valid signatures from itself. The malicious committee members can add further signatures to these objects, and then multicast all these objects. 
Since there can be unlimited number of such objects, eventually honest nodes will not have sufficient bandwidth to relay all multicast messages. Some message $x$ hence will not be properly propagated to all nodes (in time). It is possible that none of the other objects are eventually accepted by any honest nodes, while the object in $x$ is eventually accepted by some honest nodes. Since the propagation of $x$ was not properly done, the object in $x$ may not be accepted by other honest nodes, which then violates agreement/correctness.

To gain deeper insight, it helps to see why this problem does not exist when the protocol runs over a clique. In a clique, a node only plays a single role of either a committee member or a non-committee member. While a node also needs to forward messages there, 
{\em a node always first accepts an object before it forwards the object}. Throughout the execution, each node accepts at most $2$ objects, and hence sends/forwards at most $2$ messages. This is regardless of how many objects are injected by the malicious nodes.
Now with multi-hop propagation, upon receiving a certain object $x$, a relaying node $B$ {\em cannot tell whether $x$ will be accepted} (despite $B$ knowing an upper bound $d$ on the diameter of the network).\footnote{The crux here is that $B$ does not know whether $x$ can reach all nodes in time. One naive idea is for $B$ to refuse relaying $x$ when the ``residual lifespan'' of $x$ is less than $d$ rounds. This does not work because other honest nodes will do so as well, which in turn means that $B$ needs to see a  ``residual lifespan'' of at least $2d$ rounds.
This argument keeps going on without converging, from requiring $2d$ to $3d$, $4d$, and so on.}

\Paragraph{Our observation.}
Our solution to the above problem will be based on the following observation: When a node in the overlay network relays an object, while it cannot predict whether the object will be eventually accepted, the node can nevertheless determine how ``promising'' it is for the object to be accepted. Define a 
{\em push} to be the event of a node sending/forwarding/relaying a certain object $x$, together with $y$ signatures on $x$, in a certain round $t$. Intuitively, smaller $t$ and larger $y$ make it more likely for $x$ to be later accepted. More precisely, we assign each push a {\em score} of $2dy - t$ to summarize how promising it is, based on the following intuition: Roughly speaking, each signature gives the object an extra ``lifespan'' of $2d$ rounds, and an object will be accepted as long as it is received during its lifespan. The score $2dy-t$ is then the residual ``lifespan'' when the push is done in round $t$. We call a push with a higher score as a more {\em promising} push.

Now consider any node $B$, and all the pushes that $B$ has ever done. Conceptually, if all the other pushes in the network are triggered either directly or indirectly by $B$'s pushes, 
then we will have the following nice property: {\em If an object contained in a more promising push is eventually not accepted, then no objects contained in less promising pushes will ever be accepted.} Similarly, two objects contained in two pushes with the same score must have the same outcome: {\em They are either both accepted or both rejected.} (Our proofs later will formalize these properties, and also fully capture the interactions among pushes done by different nodes.)

\Paragraph{Our solution.}
With the above observation, let us proceed with the design of \tool. There may be many objects that a node $B$ needs to forward in a certain round $t$. In our design, node $B$ simply chooses the $2$ objects whose corresponding pushes would be the most promising, and forwards those $2$ objects (effectively ``materializing'' those $2$ pushes). Tie-breaking can be done arbitrarily. Note that since $t$ is fixed here, those will simply be the $2$ objects with the most number of signatures. (We nevertheless introduced the score of $2dy-t$, to facilitate later discussion.) If needed, to save storage space, $B$ can further discard all objects other than those $2$ objects. 

\Paragraph{Some intuitions.}
Section~\ref{sec:proof} will give formal correctness proofs, but we provide some quick intuitions here. First, forwarding $2$ objects (instead of one) is necessary for correctness. For example, consider the case where there would have been $2$ objects eventually accepted, if every node had materialized all possible pushes. Then forwarding only $1$ object in each round would lead to a wrong result. Second, forwarding $2$ objects in each round is also sufficient for correctness. Namely, not ``materializing'' the other less promising pushes will not cause any problem: If at least one of these two objects are not eventually accepted, then those less promising pushes would not contribute to the acceptance of any additional objects anyway. If both objects are accepted, recall from Section~\ref{sec:background3} that we no longer care about other objects, since we already have a conflict.

\subsection{Fragmentation, Delay, and Compensation}
\label{sec:design2}

\Paragraph{Avoid forwarding in every round.}
The design in Section~\ref{sec:design1} requires a node to forward 2 objects potentially in {\em every round}. For example, this may happen when a malicious broadcaster injects 2 objects in each round, with objects in later rounds being more promising. 
To further improve $\mathbb{R}$, we want to avoid forwarding objects in every round. We achieve this by using two {\em phases}. The first phase uses the design in Section~\ref{sec:design1} to broadcast the hash of the object, where a node forwards up to 2 hashes in every round (regardless of how many hashes the adversary injects). At the end of the first phase, the honest nodes will all agree on a certain hash.
The second phase uses the design in Section~\ref{sec:design1} again to broadcast the object itself. We will focus on improving the second phase, since the bandwidth bottleneck will be in the second phase.
Given the agreed-upon hash, in the second phase, each node now only needs to forward (once) a single object that matches the hash, in one single round. We call that single round as the ``busy'' round. 
Of course, each node may still have many signatures (for the object) to forward. But we leave that to Section~\ref{sec:design3}.

\Paragraph{Naive parallelism fails.}
We have explained that among the $2dm$ rounds in the second phase, each node has only one ``busy'' round. Given this, a naive attempt to improve $\mathbb{R}$ is to use simple parallelism.
Namely, we break the $l$-size object into $2dm$ fragments, build a Merkle tree with all the fragments being the leaves, and add the Merkle proof to each fragment. The first phase will now broadcast the Merkle root. The second phase would then {\em conceptually} run
$2dm$ parallel instances of the protocol in Section~\ref{sec:design1}, with one instance for each fragment.
This seems to enable each node to fully utilize all the bandwidth in the $2dm$ rounds, with one ``busy'' round from each instance. Unfortunately, a malicious broadcaster controls which round will be ``busy'' in each instance. It can thus align all the ``busy'' rounds in all the instances, so that they all occur at exactly the same time (Figure~\ref{fig_aligned_parallel}). This defeats this naive design, regardless of how we arrange the $2dm$ instances.

\begin{figure}[]
    \centering
    \includegraphics[width=\linewidth]{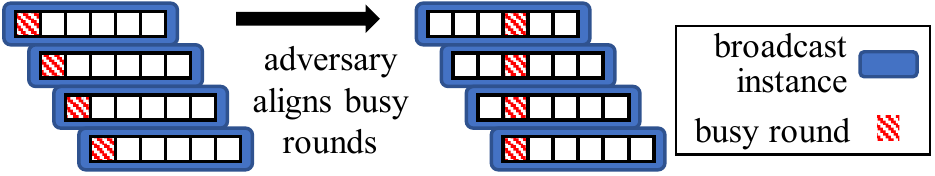}
    \vspace*{0mm}
    \caption{Naive parallelism fails.}
    \vspace*{-6mm}
    \label{fig_aligned_parallel}
\end{figure}
\Paragraph{Delay and compensation.}
Given that the adversary can choose the ``busy'' round for each node, we might just as well start all the parallel instances at the same time. 
Our first idea is that if on any node $A$, the ``busy'' rounds of two instances collide in round $t$, then $A$ will send the fragment $x_1$ in the first instance in round $t$, and delay the sending of the fragment $x_2$ in the second instance to round $t+1$. When a neighbor $B$ processes $x_2$, $B$ should compensate, and process $x_2$ as if $x_2$ were received one round earlier. Intuitively, $A$ is essentially telling $B$ that because $A$ was busy sending $x_1$ to $B$ in round $t$, the fragment $x_2$ is late by one round and is only sent in round $t+1$. Since $B$ sees that $A$ indeed sent $x_1$ in round $t$, $B$ should be willing to compensate.

If the overlay topology were a line topology, the above idea would work. In a more general topology, however, things get complicated. For example in Figure~\ref{fig_hard_compensate}, node $C_1$ sends $x_1$ to $A$ in round $t-1$, while $C_2$ sends $x_1$ and $x_2$ to $A$ in round $t$ and $t+1$, respectively. Then $A$ will send $x_1$ to $B$ in round $t$, and $x_2$ to $B$ in round $t+2$. Despite all nodes being honest in this example, $B$ sees a one-round ``gap'' between $A$'s forwarding of $x_1$ and forwarding of $x_2$. Generalizing this example can make this ``gap'' contain many rounds. In such a case, $B$ cannot be sure how much it should compensate --- in fact, since $A$ could be maliciously and intentionally add the ``gap'', $B$ cannot even decide whether to compensate at all.

\Paragraph{A classic result and its intuition.}
Before presenting our solution, we revisit a classic result~\cite{topkis85} on competing propagations in networks.
Let $x_1$ through $x_s$ be the total $s$ fragments of the object. Let us focus on the instance for $x_s$, while assuming for now that all other instances already work. (Section~\ref{sec:design3}  will show that the last instance is the key.) The propagation of $x_s$ may get delayed due to competing fragments in other instances. 

The classic result in \cite{topkis85} tells us that $x_s$ can be delayed by at most $s-1$ rounds. In particular, this is not $d\times (s-1)$ rounds, and is regardless of how the nodes prioritize different fragments during propagation. The intuition behind this result is also important. The intuition is that if $x_s$ is delayed at node $A$ for $y$ rounds, then $A$ must have been busy sending some other $y$ fragments. Once $A$ forwards those $y$ fragments before $x_s$, downstream nodes will have $y$ fewer remaining opportunities to delay $x_s$.

\Paragraph{Our solution.}
Guided by the above intuition, \tool does not have each node individually determine the amount of compensation. Instead, we use a {\em fixed amount of compensation} together with a {\em forerunner rule} during forwarding. Specifically, \tool gives a fixed compensation of $s-1$ rounds for $x_s$: Whenever any node $B$ is about to send $x_s$ in round $t$, node $B$ decides whether to accepted $x_s$, as if $x_s$ were about to be sent in round $t-(s-1)$. (If $t-(s-1) < 0$, we view it as $0$.) For example, if $B$ is a committee member, then $B$ checks whether the number of signatures on $x_s$ is at least $\frac{t-(s-1)}{2d}$.

Next, \tool 
requires nodes to follow a simple {\em forerunner rule} during forwarding: Before a node sends $x_s$, it is required to have already sent all the other $s-1$ fragments. (Those $s-1$ fragments can be sent in any ordering and in any rounds, potentially with ``gaps'' among such forwardings.) By the earlier intuition, doing so ensures that when $x_s$ is sent and when the compensation of $s-1$ is applied, all the possible delays for $x_s$ have already occurred, and there will be no further delays for $x_s$ during propagation. Note that since an honest node needs to forward all fragments anyway, the restriction from the forerunner rule has no negative effects on honest nodes. If a malicious node sends $x_s$ to an honest node $A$, without having sent all the other fragments in previous rounds, then $A$ simply ignores this message.

\begin{figure}[]
	\centering
	\includegraphics[width=\linewidth]{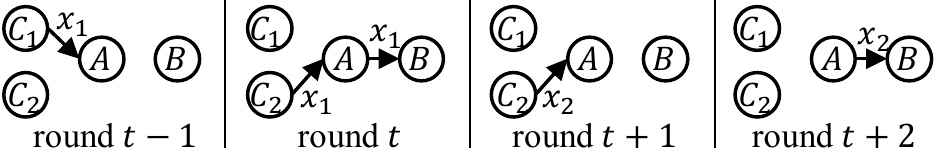}
	\vspace*{-1mm}
	\caption{$B$ sees a gap (i.e., round $t+1$) between $A$'s forwardings.}
	\label{fig_hard_compensate}
	\vspace*{-5mm}
\end{figure}

\Paragraph{Quick summary.}
Section~\ref{sec:proof} will give security analysis for the above design. As a quick summary, the classic result from \cite{topkis85} suggests that compensation of $s-1$ rounds will always be sufficient. The intuition behind this classic result, together with our forerunner rule, roughly suggests that if $x_s$ has already experienced a delay of $s-1$ round by the time that a node $B$ sends $x_s$, then $x_s$ will not experience further delays in downstream honest nodes. This ultimately implies that the key invariant from the Dolev-Strong protocol~\cite{dolev83} still holds: When a node accepts $x_s$, it knows that it can make all other nodes accept $x_s$ within a $2d$ rounds.

\subsection{Forwarding Signatures and Combining the Two Phases}
\label{sec:design3}

\Paragraph{Forwarding signatures.}
Section~\ref{sec:design2} ignored the overhead of sending the signatures. 
To minimize such overhead, 
\tool only uses signatures on the last fragment $x_s$. The other fragments do not carry signatures, and there is no notion of acceptance for each such fragment individually. The entire object (i.e., all its $s$ fragments) is accepted iff $x_s$ is accepted. To intuitively see why this works, note that by our forerunner rule, if $x_s$ is accepted, then the node must have previously sent (and hence seen) all the other $s-1$ fragments. Thus if a node accepts $x_s$, it must be able to reconstruct the object from all the fragments. Furthermore, if all the honest nodes agree on whether $x_s$ is accepted, they must also have agreement on whether the object is accepted. 
As a further optimization, since $x_s$ is the only fragment carrying signatures, we want to make $x_s$ as small as possible. To do so, the broadcaster simply chooses a random nonce as $x_s$, and the object is now split into only $s-1$ fragments. 

\label{sec:design4}
\Paragraph{Running the two phases in parallel.}
The design in Section~\ref{sec:design2} requires two sequential phases: the first phase for the Merkle root and the second phase for the object itself. To further improve performance, we next explain how to run these two phases in parallel, using the following two modifications.

First, a node in the second phase needs to determine whether a fragment is a leave of the Merkle tree with root $r$, where $r$ is agreed upon at the end of the first phase.
When the two phases run in parallel, such determination cannot be easily made anymore. But recall from Section~\ref{sec:design1} that every node $A$ assigns a score to every push that it has ever done. Such a score captures how promising the push is. Now with the two phases running concurrently, 
in the (concurrent) second phase, node $A$ simply uses the Merkle root $r^A$ contained in its most promising push done so far in the first phase, as its current {\em guess} for $r$. 
Our later proof will show that using such a guess suffices to ensure the correctness of the protocol.

Second, let $t^A$ be the round during which $A$ is about to send the last fragment $x_s$ in the second phase. Previously in Section~\ref{sec:design2}, $A$ would decide whether to accept $x_s$ based on the value of $t^A$. 
Now that the two phases run in parallel, we need to adjust this part as well. Specifically, let $t^A_{\textnormal{root}}$ be the round during which $A$ accepts $r$ in the first phase, and define $t^A_{\textnormal{frag}} = \max(t^A, t^A_{\textnormal{root}}+s-1)$. When deciding whether to accept $x_s$, node $A$ will make the decision as if $x_s$ were sent in round $t_{\textnormal{frag}}^A$  (instead of in round $t^A$).
The exact reasoning behind this $t_{\textnormal{frag}}^A$ term is slightly complex. For the lack of space, instead of going through a lengthy example here, we directly prove the correctness of such a design later.

\section{Complete Pseudocode for \tool}
\label{sec:pseudocode}

\Paragraph{The PoS setting.}
Section~\ref{sec:design} assumed a permissioned setting. The actual \tool protocol is designed for a
permissionless PoS setting. With PoS, each node holds some coins
(i.e., stakes). For each \tool invocation, Section~\ref{sec:blockchain} later will choose $m$
random coins among all these coins. The nodes holding those coins then
become committee members in \tool. A node $B$ may hold $x\ge 1$
chosen coins.
In such a case, 
$B$'s signature will be viewed as being equivalent to $x$ signatures from
$x$ different committee members. We also call $x$ as the {\em weight}
of $B$. 
Among the $m$ chosen coins, the node holding the first chosen coin will further 
be the broadcaster in \tool. The information
regarding which coins are chosen will be public --- specifically, they are chosen by some random beacon, which is periodically computed and released. Hence all parties
know the public keys (but not necessarily IP addresses) of all the committee members, each time before \tool is 
invoked.\footnote{We will explain later that each epoch in \codename computes a beacon to select the committees in the next epoch. Hence our design allows a {\em mildly-adaptive} adversary as in \cite{rapidchain, snow-white,ouroboros} --- namely, if it takes multiple epochs for the adversary to adaptively corrupt nodes, then the adversary will not be able to cherry-pick the committee members to corrupt, after seeing the beacon and before the committee members have done their work.}

\Paragraph{Signature aggregation.}
\tool uses signature aggregation to reduce signature size, as an optimization. 
One suitable signature aggregation scheme is the $\mathpzc{MSP\text{-}pop}$ scheme using BLS381, which gives aggregate signatures of size only $96$ bytes~\cite{multisig-bdn18}.
The $\mathpzc{MSP\text{-}pop}$ scheme requires certain public parameters, which can easily be published in the genesis block of \codename. Each node can generate public keys independently and non-interactively, as and when needed, based on these public parameters. $\mathpzc{MSP\text{-}pop}$ additionally requires a proof-of-possession for each public key.
In \codename, we simply require a node to add a transaction containing this proof to the blockchain, before it is allowed to be a committee member.
In each invocation of \tool, the possible signers are all the $m$ committee members for that invocation.
Hence for each aggregate signature, an $m$-bit vector suffices to indicate which of the $m$ members are signers. 

Consider any Merkle root $x$ or fragment $x$. In our pseudo-code, the
set ${\tt all\_sig}$ keeps track of all the aggregate signatures seen
by a node so far. Note that ${\tt all\_sig}$ may contain multiple
aggregate signatures for $x$, since we do not combine multiple
aggregate signatures into one. We use $\sigma(x)$ to denote the aggregate signature
for $x$ whose signers have the largest total weight, with arbitrary
tie-breaking, among all aggregate signatures in ${\tt
  all\_sig}$. If there is no aggregate signature for
$x$ in ${\tt all\_sig}$, we define $\sigma(x)=\emptyset$.
We use $|\sigma(x)|$ to denote the total weight of the signers in
$\sigma(x)$. We use $\sigma(x).{\tt add\_my\_sig}()$ to
denote the aggregate signature obtained by adding the invoking node's
signature to $\sigma(x)$.
If the invoking node is already a signer in
$\sigma(x)$, then $\sigma(x).{\tt add\_my\_sig}() = \sigma(x)$.

\setlength{\textfloatsep}{8pt}\setlength{\floatsep}{5pt}\begin{algorithm}[t]
	\small
	\caption{\tool (\textbf{Parameters}: $m, d, s$)}
	\label{alg_top}
	\hspace*{\algorithmicindent} 
\begin{algorithmic}[1]
		\State ${\tt all\_root} \leftarrow \emptyset$; // all received (Merkle) roots
		\State ${\tt all\_push} \leftarrow \emptyset$; // all pushes done so far for roots
		\State ${\tt root\_accepted} \leftarrow \emptyset$; // roots accepted so far
		\State $t_{\textnormal{root}}\leftarrow \infty$; // round number when first root accepted 
		\State ${\tt all\_frag} \leftarrow \emptyset$; // received fragments
		\State $\tt frag\_accepted \leftarrow false$; // last fragment has been accepted?
		\State ${\tt all\_sig} \leftarrow \emptyset$; // received signatures on roots and fragments
		\State
		\If {I am the broadcaster}
		\State break the object (to be broadcast) into $s-1$ fragments;
		\State pick a random nonce as the last fragment (i.e., $s$th fragment); 
		\State let $r$ be the Merkle root of all these $s$ fragments;
		\State add the Merkle proof into each fragment;
		\State ${\tt all\_root} \leftarrow {\tt all\_root} \cup \{r\}$;
		\State ${\tt all\_frag} \leftarrow {\tt all\_frag} \cup \{\text{the }s\text{ fragments}\}$;
\EndIf
		\State
		\For{$t$ from $0$ to $2dm+s-1$ (both inclusive)}\label{line_forloop}
		\State receive messages from all neighbors;\label{line_receive}
		\State discard those received Merkle roots whose aggregate signatures do not contain the broascaster
               as a signer; \label{line_broadcast_sign}
		\State add received Merkle roots to $\tt all\_root$; 
		\State add received fragments to $\tt all\_frag$;
		\State add received aggregate signatures to $\tt all\_sig$;\label{line_receive_end}
		\State ForwardMerkleRoot(); ForwardFragment();\label{line_invoke_sub}
		\State wait until the current round $t$ ends;\label{line_round_wait}
		\EndFor\label{line_forloop_end}
		\State
		\If{$(|{\tt root\_accepted}| = 1) \wedge ({\tt frag\_accepted} = {\tt true})$}  \label{line_final}
		\State \Return the object by combining the fragments (in $\tt all\_frag$) that correspond to the (single) Merkle root in ${\tt root\_accepted}$; // we will prove that there are exactly $s$ such fragments 
		\label{line_returnobj}
		\Else \hspace*{0mm} \Return $\bot$;
		\EndIf
	\end{algorithmic}
\end{algorithm}

\begin{algorithm}
	\small
	\caption{ForwardMerkleRoot()}\label{alg_root}
	\begin{algorithmic}[1]
		\setcounter{ALG@line}{31}
		\State {\bf if} $|{\tt all\_root}|\le 1$ {\bf then} ${\tt top\_root}\leftarrow {\tt all\_root}$;\label{line_root_1root}
		\State {\bf else} 
${\tt top\_root} \leftarrow \{r_1, r_2\}$ such that $|\sigma(r_1)|\ge |\sigma(r_2)|\ge |\sigma(r)|$ for all $r \in {\tt all\_root}$; // tie-breaking can be done arbitrarily
		\label{line_root_2roots}
		\For{each $r \in {\tt top\_root}$}\label{line_root_pick}
				\If{(I am in committee) and ($2d|\sigma(r)|\geq t$)}\label{line_root_comm}
					\State ${\tt all\_sig} \leftarrow {\tt all\_sig} \cup \{\sigma(r).{\tt add\_my\_sig()}\}$ \label{line_root_addsig}; \State ${\tt root\_accepted}\leftarrow {\tt root\_accepted} \cup \{r\}$;
					\label{line_committee_accept_root}
					\State $t_{\textnormal{root}} \leftarrow \min(t_{\textnormal{root}},t)$;\label{line_root_t0_comm}
				\EndIf
				\If{(I am not in committee) and ($2d|\sigma(r)|\geq t+d$)}\label{line_root_noncomm}
					\State ${\tt root\_accepted}\leftarrow {\tt root\_accepted}\cup \{r\}$;
					\label{line_noncommittee_accept_root}
					\State $t_{\textnormal{root}} \leftarrow \min(t_{\textnormal{root}},t)$;\label{line_root_t0_noncomm}
				\EndIf
			\State send $r$ and $\sigma(r)$ to all my neighbors;\label{line_root_push}
			\State let $p$ be the push corresponding to the above send;
			\State $p.{\tt score}\leftarrow 2d|\sigma(r)|-t$;\label{line_root_calculate}
			\State ${\tt all\_push} \leftarrow {\tt all\_push} \cup \{p\}$;
		\EndFor
\end{algorithmic}
\end{algorithm}

\begin{algorithm}
	\small
	\caption{ForwardFragment()}\label{alg_frag}
	\begin{algorithmic}[1]
		\setcounter{ALG@line}{48}
		\State {\bf if} ${\tt all\_push} = \emptyset$ {\bf then return}; \label{line_frag_nopush}
		\State let $p \in {\tt all\_push}$ be the push with largest $p.{\tt score}$; // tie-breaking can be done arbitrarily \label{line_frag_pick_root}
		\State let $x_1$ through $x_s$ denote the $s$ fragments corresponding to the Merkle root in $p$; // I may or may not have received all of them
		\State
		\If{(there exists any $i\in [1,s-1]$ such that $x_i\in {\tt all\_frag}$ and I have not forwarded $x_i$ before)}\label{line_frag_allfrag}
		\State pick any such $i$ and send $x_i$ to all my neighbors; \label{line_send_xi}
		\State {\bf return};
		\EndIf\label{line_frag_allfrag_end}
		\State
		\If{($x_i\in {\tt all\_frag}$ for all $i\in [1,s])$}\label{line_frag_receive_last}
		\State $t_{\textnormal{frag}} \leftarrow \max(t, t_{\textnormal{root}}+s-1)$;\label{line_frag_t2}
		\If{(I am in committee) and ($2d|\sigma(x_s)|\geq t_{\textnormal{frag}}-(s-1)$)}\label{line_frag_comm}
		\State ${\tt all\_sig} \leftarrow {\tt all\_sig}\,\,\cup$ $\{\sigma(x_s).{\tt add\_my\_sig}()\}$;\label{line_frag_addsig}\State ${\tt frag\_accepted} \leftarrow {\tt true}$;\label{line_frag_comm_set}
		\EndIf
		\If{(I am not in committee) and ($2d|\sigma(x_s)|\geq t_{\textnormal{frag}}-(s-1)+d$)}\label{line_frag_noncomm}
		\State ${\tt frag\_accepted} \leftarrow {\tt true}$;\label{line_frag_noncomm_set}
		\EndIf
		\State send $x_s$ and $\sigma( x_s)$ to all my neighbors; \label{line_frag_send}
		\EndIf
	\end{algorithmic}
\end{algorithm}

\Paragraph{Algorithm 1.}
Algorithm 1 is the main algorithm for \tool, run by every node in the system.
\tool has total $2dm+s$ rounds (Line~\ref{line_forloop} to \ref{line_forloop_end}). Here $2dm$ follows from the discussion in Section~\ref{sec:background4}, while the $s$ rounds comes from the delay/compensation design in Section~\ref{sec:design2}. Recall from Section~\ref{sec:model} that each node uses its local clock to keep track of the beginning of the execution (not explicitly shown in the pseudo-code) as well as the progress of each round (Line~\ref{line_round_wait}).

The two phases, one for the Merkle root and one for the object itself, run in parallel by the design in Section~\ref{sec:design4}. In each round, a node first adds the various received roots/fragments/signatures into the corresponding sets (Line~\ref{line_receive} to \ref{line_receive_end}). Next Line~\ref{line_invoke_sub} invokes ForwardMerkleRoot() and ForwardFragment() to do the processing for the first and second phase, respectively. 

After all these $2dm+s$ rounds, a node outputs a non-$\bot$ object iff i) the first phase has accepted exactly one Merkle root $r$, and ii) the second phase has accepted the last fragment corresponding to this Merkle root $r$.

\Paragraph{Algorithm 2.}
Algorithm 2 largely follows the design in Section~\ref{sec:design1}. In particular, Line~\ref{line_root_2roots} chooses two Merkle roots with the largest total weight of signers, and Line~\ref{line_root_calculate} computes the score of the push. 
At Line~\ref{line_root_comm}, a committee member accepts a root $r$ if the total weight of signers is at least $\lceil \frac{t}{2d} \rceil$. This matches the intuition in Section~\ref{sec:background2} and \ref{sec:design1}, since each round in Dolev-Strong protocol~\cite{dolev83} corresponds to $2$ rounds in Chan et al.'s protocol~\cite{chan20} (under clique setting), which in turn map to $2d$ rounds in OverlayBB. 
Similarly, a non-committee member accepts a root $r$ if the total weight of signers is at least $\lceil \frac{t+d}{2d} \rceil$ --- this simply means that within $d$ rounds, the root $r$ will reach some committee member, and will be accepted by that committee member. 

\Paragraph{Algorithm 3.}
Algorithm 3 corresponds to (one round of) the second phase in Section~\ref{sec:design2}. Section~\ref{sec:design2} explained that {\em conceptually}, the second phase uses one instance for each fragment, with total $s$ instances. But in each round, a node only sends message for at most one instance. Algorithm 3 chooses that instance (implicitly) at Line~\ref{line_send_xi} and \ref{line_frag_receive_last}, and then processes {\em only} that single instance. Hence Algorithm 3 remains single-threaded, despite that it actually implements $s$ parallel instances.

Line~\ref{line_frag_pick_root} follows the design in Section~\ref{sec:design4}, and uses the Merkle root contained in the most promising push as a guess for the final accepted root.
Line~\ref{line_frag_allfrag} to \ref{line_frag_allfrag_end} follow the forerunner rule in  Section~\ref{sec:design2}. Line~\ref{line_frag_t2} computes $t_{\textnormal{frag}}$ as discussed in Section~\ref{sec:design4}. Line~\ref{line_frag_comm} and \ref{line_frag_noncomm} check whether to accept $x_s$, based on $|\sigma(x_s)|$, $t_{\textnormal{frag}}$, and the $s-1$ compensation as discussed in Section~\ref{sec:design2}. The actual decision rule is similar to Line~\ref{line_root_comm} and Line~\ref{line_root_noncomm}.

\section{From Byzantine Broadcast to Blockchain}
\label{sec:blockchain}

So far we have presented our byzantine broadcast protocol, \tool. We now explain how to use \tool to build our blockchain, \codename. 

\Paragraph{Basic design.}
While largely neglected in the literature, blockchains can be relatively easily built from byzantine broadcast, in the following way. In  a blockchain protocol, every node aims to maintain an append-only sequence of blocks, and all the sequences on all the honest nodes need to be consistent with each other. 
For convenience, imagine that there is a sequence of {\em slots}, which initially are all empty. The nodes in the system invoke \tool periodically (e.g., every $98$ seconds), and a node uses the return value from the $i$th invocation of \tool as the block for the $i$th slot. 
For each invocation, the broadcaster is chosen randomly, who assembles a block and then uses \tool to disseminate that block to all nodes. We say that a block/slot is {\em confirmed} if its corresponding \tool invocation has ended. Note that since each \tool invocation can take much longer than $98$ seconds, the $(i+1)$th invocation will start before the $i$th invocation ends. This effectively results in {\em pipelined invocations}, and at any point of time, there can be many 
active \tool invocations. 
All these pipelined invocations can be implemented efficiently:
Our actual implementation will simply use a {\em single thread} to loop through all the pipelined invocations, and process them one by 
one.
We also stagger the round starting time of all these invocations, so that invocations near the end of the processing loop start their rounds a bit later.

The above basic framework is already used in Pass and Shi~\cite{pass2018thunderella}, which describes a theoretical design of a blockchain using the Dolev-Strong protocol~\cite{dolev83} (instead of \tool).
\codename also follows this basic framework, but there are several practical issues we need to overcome, as following.

\Paragraph{Choosing broadcaster/committee.}
We use two independent hash functions, ${\tt hash}_1$ and ${\tt hash}_2$, in \codename. 
The execution of \codename is divided into {\em epochs} (e.g., 1 epoch $=$ 1 day). 
In each epoch $i-1$, the nodes compute (explained later) a fresh public random beacon, denoted as ${\tt beacon}_i$, to be used in epoch $i$. 
Recall from Section~\ref{sec:model} that the {\em genesis block} contains an unbiased random beacon to be used in the very first epoch. Hence the genesis block bootstraps this sequential process of beacon generation.

We say that a slot/block is {\em in an epoch} if the starting time of the corresponding \tool invocation is in that epoch. Note that the ending time may be in the next epoch. 
Let $y$ be the last block that has been confirmed by the beginning of epoch $i-1$ (i.e., before the computation of ${\tt beacon}_i$ starts). Let the coin distribution $\mathbb{D}$ (i.e., which nodes hold which coins) be 
the coin distribution immediately after block $y$ (i.e., when we apply all the transactions in blocks 1 through $y$). For the $k$th slot in epoch $i$, every honest node uses ${\tt hash}_1(k|{\tt beacon}_i)$ as randomness to 
select $m$ coins (with replacement) from $\mathbb{D}$. The holders (in $\mathbb{D}$) of these coins then become the committee for the \tool invocation corresponding to that slot. The holder of the first coin selected will be the broadcaster. Since ${\tt beacon}_i$, $k$, and $\mathbb{D}$ are all public information in epoch $i$, all honest nodes will select the same broadcaster/committee, if their have the same sequences of blocks prior to epoch $i$.

\Paragraph{Generating beacons: Overview.}
Beacon generation is a central issue in PoS blockchains, and there have been a number of prior approaches~\cite{snow-white,algorand,ouroboros}. Some of these~\cite{algorand,ouroboros} do not work well under malicious majority. We build upon the approach in \cite{snow-white}. Roughly speaking, they~\cite{snow-white} observe that the beacon is eventually only used to select a committee. Assuming that a random oracle is used to select the committee based on the beacon, the committee will be bad (e.g., having no honest committee member) with only exponentially small probability. Hence a computationally-bounded adversary simply will have a hard time finding a bad beacon, even if it can choose any beacon it wants.

Directly adopting this idea in \codename does not lead to a practical solution, since the number of beacons the adversary can try is still huge. To make it work, we use a simple idea of {\em weak Proof-of-Work} ({\em weak PoW}), so that generating a valid beacon takes some computational effort.
Recall that Section~\ref{sec:model} assumed that the adversary's computational power is at most $100$ times of the computational power of the honest nodes. 

\setlength{\textfloatsep}{20pt}\begin{figure}[]
    \centering
    \includegraphics[width=\linewidth]{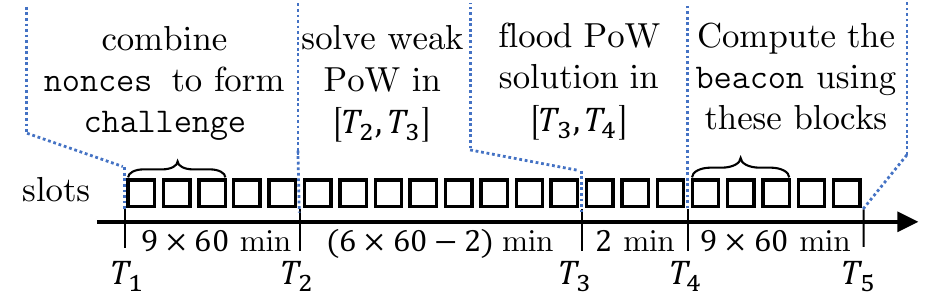}
    \vspace*{-2mm}
    \caption{Generating the beacon in an epoch. The number of slots in each portion is not to scale.}
    \label{fig_beacon_gen}
    \vspace*{-5mm}
\end{figure}

\Paragraph{Generating beacons: Details.}
To facilitate beacon generation, each block in \codename contains two additional fields: ${\tt nonce}$ and ${\tt candidate}$. The ${\tt nonce}$ field is just some uniformly random bits locally generated by the broadcaster, who is also the creator, of that block. A block is called an {\em honest block} if its broadcaster/creator is an honest node. 

Recall that the nodes generate ${\tt beacon}_i$ during epoch $i-1$. Let $T_1$ and $T_5$ be the start and end time, respectively, of epoch $i-1$  (Figure~\ref{fig_beacon_gen}). Let $T_2$ be the time when the first $\tau$ slots in epoch $i-1$ have been confirmed. At time $T_2$, all the ${\tt nonce}$ values in these $\tau$ blocks are concatenated (not XOR-ed), and used as the fresh ${\tt challenge}$ for the weak PoW in this epoch. Here $\tau$ is chosen such that with high probability, there is at least one honest block (and hence one honest
${\tt nonce}$) among those $\tau$ blocks. This ensures that the adversary cannot pre-compute PoW solutions before the beginning of epoch $i-1$.

The honest nodes will try solving the weak PoW, starting from time $T_2$ and until time $T_3$, where $T_3$ can be any value no larger than $T_4-d\delta$. Here $T_4$ is the latest time such that there are still $\tau$ slots (called {\em candidate-holding slots}) whose \tool invocations have not yet started at time $T_4$, but will end by time $T_5$. 
To solve the weak PoW, a node needs to find $x$ such that ${\tt hash}_2({\tt challenge}|x)$ has a certain number of leading zeroes. We also call such $x$ as a PoW solution.

At time $T_3$, every node will flood/multicast whatever PoW solution (if any) it has found. To avoid unnecessary bandwidth consumption, each node only sends/relays the {\em very first} PoW solution it finds/receives, and {\em ignores} all other PoW solutions. Since all the honest nodes form a connected component, if they collectively find at least one PoW solution by $T_3$, then every honest node $B$ must see some PoW solution by $T_4$. But different honest nodes may see different PoW solutions.

The period from $T_4$ to $T_5$ serves to enable the honest nodes to agree on one PoW solution.
To achieve this, from time $T_4$ to $T_5$, whenever a node is chosen as the broadcaster, it sets the ${\tt candidate}$ field in its block to be the single PoW solution that it previously sent/relayed, 
or ${\tt null}$ if it does not have any. 
Finally, at time $T_5$, a node examines all the $\tau$ candidate-holding slots, and picks the first block with a ${\tt candidate}$ value that is not ${\tt null}$. It then uses ${\tt hash}_2({\tt challenge}|{\tt candidate})$ as ${\tt beacon}_i$. Note that most likely, this ${\tt candidate}$ is from a malicious block and is set by the adversary. This is not a problem --- all we need is that i) the adversary do not have too many candidates to choose from, and ii) the honest nodes agree on this ${\tt candidate}$. 

If all the $\tau$ candidate-holding slots have ${\tt candidate} = {\tt null}$, then a node will set ${\tt beacon}_i$ to be ${\tt beacon}_{i-1}$, which means that the system simply reuses the old beacon and the corresponding old coin distribution $\mathbb{D}$. Note that this can only occur when either there is no honest block in the $\tau$ candidate-holding slots (whose probability can be tuned by adjusting $\tau$), or the honest nodes have found no PoW solution. 
Our analysis next will fully take into account all such possibilities.

\section{Security Analysis}
\label{sec:proof}

This section analyses the security guarantees, or more specifically, {\em safety} and {\em liveness/throughput}, of \codename. 

\subsection{Safety of \codename}
Recall that in \codename, each node maintains an append-only sequence of blocks. The $i$th block is simply the return value from the $i$th invocation of \tool. A node invokes \tool periodically (e.g., every 98 seconds), and each invocation takes the same amount of time to complete. 
Hence a node adds blocks, one by one, to the sequence.

{\em Safety} of \codename essentially means that for each $i\ge 1$ and after the $i$th invocation of \tool returns, the $i$th block on all honest nodes should always be the same.
This is also sometimes called the {\em consistency} or {\em agreement} property of the blockchain. To eventually prove such safety guarantee of \codename, the following lemma first summarizes the properties of \tool, whose proof is deferred to Appendix~\ref{app:theagreement}:
\begin{theorem}[{restate=[name=Restated]objagree}]
	\label{thm_agreement}
	{\bf [guarantees of \tool{}]} 
	In Algorithm~\ref{alg_top}, if the committee has at least one honest member, then
	\begin{itemize}
		\item All honest nodes must return the same object.
		\item If the broadcaster is honest, then all honest nodes must return the object broadcast by the broadcaster.
	\end{itemize} 
Finally, regardless of the committee, Algorithm~\ref{alg_top} always returns within $2dm+s$ rounds.
\end{theorem}
Part of Theorem~\ref{thm_agreement} requires the committee to contain some honest member. Consider any slot in epoch $i$. Recall that the committee for that slot is chosen using ${\tt beacon}_i$. This ${\tt beacon}_i$ is generated in epoch $i-1$, and may be biased and influenced by the adversary: i) the adversary may find multiple PoW solutions in epoch $i-1$, and cherry-pick the one that it likes; ii) if the $\tau$ broadcasters in the first $\tau$ slots of epoch $i-1$ are all malicious, then the adversary can predict the PoW challenge before time $T_1$, and can pre-compute many PoW solutions; iii) if the honest nodes fail to find any PoW solution in epoch $i-1$ or if the $\tau$ broadcasters in the $\tau$ candidate-holding slots in epoch $i-1$ are all malicious,
then ${\tt beacon}_{i-1}$ may be reused as ${\tt beacon}_i$, and ${\tt beacon}_{i-1}$ may already be biased; iv) if the committee for some slot in epoch $i-1$ contains no honest members, then the honest nodes may not even agree on the PoW challenges and on what ${\tt beacon}_i$ is.

We will later reason about the probabilities of various random events, such as whether the committee in Theorem~\ref{thm_agreement} contains some honest member. 
The adversary may influence such probabilities, by for example, biasing the beacons as explained above. The amount of such influence will depend on what strategy the adversary uses. We will carefully ensure that all our analyses (e.g., regarding the probabilities) hold, even under {\em the worst-case adversary that uses the optimal strategy}. In particular, our analyses will not make claims such as $\Pr[X]=y$, but only make claims such as $\Pr[X]\le y$. This just means that while $\Pr[X]$ may be different under different strategies of the adversary, it can never be above $y$. 

We now introduce some random variables.
Consider all the slots in the blockchain. Let $\rho$ be the number of slots in each epoch.
For all integer $j \in [1, \infty)$ and all $\lambda\ge 1$, let random variable $\mathcal{Z}_\lambda(j)$ denote the event that all of the following events happen in the execution of \codename:
\begin{itemize}
	\item For each slot $j'$ where $1\le j' \le j$, the committee for that slot contains at least one honest member.
	\item For each epoch $i'$ where $1\le i'\le \lceil\frac{j+1}{\rho}\rceil - 1$, no more than $\lambda$ different PoW solutions are seen by honest nodes
in epoch $i'$.
\end{itemize}
For all $\lambda\ge 1$, define $\mathcal{Z}_\lambda(0)$ to be an event that always occurs.

Roughly speaking, $\mathcal{Z}_\lambda(j)$ means that the execution is ``good'' up to slot $j$.
The first part in $\mathcal{Z}_\lambda(j)$ corresponds to the requirement in Theorem~\ref{thm_agreement}, and 
the second part serves to facilitate later reasoning about $\Pr[\mathcal{Z}_\lambda(j)]$ via a recursion.
In this second part, the $\lambda$ solutions ``seen by honest nodes'' can be i) PoW solutions for epoch $i'$ found by honest nodes, ii) PoW solutions for epoch $i'$ found by malicious nodes in epoch $i'$, and iii) PoW solutions for epoch $i'$ found by malicious nodes before epoch $i'$ started (if the PoW challenge is not fresh). 

We now formally state the safety guarantee of \codename:
\begin{theorem}
	\label{the:bcubesafety}
	{\bf [safety guaranteed in ``good'' execution]} 
	For any given $\lambda\ge 1$ and $j\ge 1$, if $\mathcal{Z}_\lambda(j)$ occurs, then for each $j'$ where $1\le j' \le j$, all honest nodes in \codename must always have the same block in slot $j'$ once the \tool invocation for slot $j'$ has completed.
\end{theorem}
\begin{proof}
	Consider any given $j'$ where $1\le j' \le j$.
	Then $\mathcal{Z}_\lambda(j)$ occurring means that the committee for slot $j'$ has some honest member. For any given node, the $j'$-th block in its blockchain is simply the return value of the \tool invocation on that node for the $j'$-th slot. By Theorem~\ref{thm_agreement}, such return value must be the same on all honest nodes.
\end{proof}

Theorem~\ref{the:bcubesafety} guarantees the safety of \codename in ``good'' executions, but does not tell us the likelihood of the execution being ``good''. Theorem~\ref{the:extension} 
next shows that conditioned upon the execution being ``good'' up to slot $j-1$, with high probability, it continues to be ``good'' up to slot $j$. 
Theorem~\ref{the:extension} is based on the following parameterization of \codename: We set $T_1$ through $T_5$ to match the respective durations in Figure~\ref{fig_beacon_gen}, and we set the weak PoW difficulty so that the honest nodes on expectation obtain two PoW solutions from $T_2$ to $T_3$. Changing these parameters will only affect the two constants ``$0.86$'' and ``$807$'' in the theorem. Also, Theorem~\ref{the:extension} assumes that the adversary cannot adaptively corrupt honest nodes. Appendix~\ref{app:theextension} will explain that the negative effect of adaptive corruption is easily bounded, as long as the adaptivity is sufficiently ``mild''. 

\begin{theorem}[{restate=[name=Restated]theextension}]
	\label{the:extension}
	{\bf [``good'' execution occurs w.h.p.]}
	Consider any constant $f \le 0.99$, and any positive integers $\lambda$ and $j$.  If $\Pr[\mathcal{Z}_\lambda(j-1)]> 0.9$, then conditioned upon $\mathcal{Z}_\lambda(j-1)$, we must have:\footnote{We use Poisson distribution to approximate binomial distributions here. } \vspace*{1mm} \\ 
 \hspace*{2mm}$\Pr[\mathcal{Z}_\lambda(j)]\,\ge\, 1-\frac{\lambda f^m}{0.9(0.86-f^\tau)} - \frac{\lambda f^\tau}{0.9(0.86-f^\tau)} - {\tt Pois}(807, \lambda)$\vspace*{2mm}\\ 
 \hspace*{18mm}$=\, 1 - \lambda e^{-\Omega(m)} - \lambda e^{-\Omega(\tau)} - e^{-\Omega(\lambda)}$\vspace*{1mm}\\
 Here ${\tt Pois}(807,\lambda)$ is defined to be $\Pr[X>\lambda]$, where $X$ follows a Poisson distribution with mean $807$.
\end{theorem}
\begin{proof} See Appendix~\ref{app:theextension}. \end{proof}

\noindent 
Asymptotically, the error probability in the above theorem is exponentially small with respect to $m$, $\tau$, and $\lambda$, which can all be viewed as security parameters.
As a concrete example, with $f= 0.7$, $\lambda = 1000$, $m \ge 79$, and $\tau \ge 91$, the above theorem gives\footnote{Conceptually, Theorem~\ref{the:bcubesafety} focuses on the error probability of a given committee (i.e., for the $j$th slot) in \codename. This is consistent with other analysis in the literature~\cite{elastico,omniledger,rapidchain,byzcoin}. If needed, one can easily translate such guarantees to the entire execution.}
$\Pr[\mathcal{Z}_\lambda(j)] \ge 1-2^{-30}$. Hence we use a committee size of $m=80$ in our later experiments when $f=0.7$.

\subsection{Liveness and Throughput of \codename}

For any given slot, {\em liveness} of \codename means that \codename should always eventually confirm a block for that slot. {\em Throughput} simply equals block size times the average number of blocks confirmed per second. In some sense, throughput captures the ``rate of liveness'', in terms of the number of bits confirmed per second. Note that each slot in \codename has a corresponding \tool invocation, which starts at a pre-determined time. The following theorem shows that once the invocation starts, within some well-defined time, we will have a confirmed block in that slot:
\begin{theorem}
	{\bf [liveness guaranteed]}
	At most $2dm+s$ rounds (or $(2dm+s)\delta$ time with $\delta$ being the round duration) after the start of the corresponding \tool invocation for a given slot in \codename, all honest nodes in \codename must have a confirmed block in that slot. 
\end{theorem}
\begin{proof}
	Trivially follows from the fact that Algorithm 1 has exactly $2dm+s$ rounds.
\end{proof}

Due to space constraint, we defer our throughput analysis of \codename to Appendix~\ref{app:throughputanalysis}. 
Appendix~\ref{app:throughputanalysis} first derives an upper bound on the total number of bits that each honest node needs to send in each round. This upper bound will hold under {\em all possible strategies of the adversary and all possible randomness outcomes}. Using this upper bound, 
Appendix~\ref{app:throughputanalysis} then shows that, under practical parameters, \codename has a throughput of $\mathbb{T} \approx \frac{\mathbb{B}}{2w} = \Theta(\frac{\mathbb{B}}{w})$ and a TTB ratio of $\mathbb{R} \approx \Theta(\frac{1}{w})$.

 \section{Implementation and Experimental Results}
\label{sec:exp}

\Paragraph{Implementation.}
We have implemented \codename in Go and using TCP, except the following parts that have no effects on our experimental results: Since beacon generation from the weak PoW takes one epoch (e.g., one day), we did not implement the weak PoW or propagate the PoW solutions. (Propagating the PoW solutions has negligible cost, since each node only sends/relays one 20-byte PoW solution in each epoch.) We instead directly inject a random beacon. We still properly determine various parameters, such as committee size, based on our weak PoW design.
We did not implement transactions, and we fill each block with random bits. There is no stake transfer, and each node always holds one stake (coin). Finally, we will run up to $500$ \codename nodes on each physical machine. Due to CPU constraint, we did not implement aggregate signature signing/validation, and also did not implement secure hash function. We replace all of these with dummy functions. Appendix~\ref{appendix:crypto} 
will show, via a careful calculation, that {\em regardless of the strategy of the adversary and regardless of what messages the malicious nodes may send to the honest nodes}, under all settings in this section, 
in every second each honest \codename node only needs to do at most $152$ aggregate signature signing/validation operations, and at most $610$ secure hashes (for Merkle proof verification). 
Similarly, due to memory constraint, 
the $500$ \codename nodes on the same machine are implemented as separate threads in one Go process, instead of as $500$ separate Go processes. Of course, these threads do not interact with each other via the shared heap space.

\Paragraph{Experimental settings.}
We run our experiments on $21$ high-end PCs, each with $10$Gbps bandwidth, in a local-area network.
The first PC runs a single \codename node (with the maximum degree of $42$ --- see later). Each of the remaining $20$ PCs run $500$ \codename nodes (with the last PC running $499$ nodes), so that each \codename node has about $20$Mbps bandwidth. Altogether, this gives us total 10000 \codename nodes. Running one \codename node on the first PC allows us to directly measure the total network traffic on the Ethernet interface of that PC in every second, using the Linux bandwidth monitoring tool ${\tt bmon}$. Our measurement results in Appendix~\ref{app:sanity} confirm that a \codename node (even with the maximum degree of $42$) indeed never uses more than $20$Mbps bandwidth.

We construct the overlay topology in a similar way as in \cite{algorand,ohie20}: Each node $A$ keeps choosing random nodes to establish (undirected) edges to, until it manages to establish $20$ edges. To prevent $A$ from forming edges to the other nodes on the same machine as $A$ and hence bypassing the network, the random nodes are chosen from all the nodes on the other machines. 
To avoid having too many neighbors, each node stops accepting new edges after it has accepted $22$ edges from other nodes. Hence the node degrees range from $20$ to $42$, with the average being $40$. We set $\delta_1 = 10$s, $\delta_2 = 2$s, and $\delta = 12$s. We assume that all edges are good in our experiments. 
While we do not explicitly emulate wide-area message propagation delay, we 
expect such delay to be typically well below our $\delta_1$ value of $10$ seconds. (The messages in our experiments always have size no larger than $10$KB.) We observe that with the above construction, the honest subgraph typically has diameter of no more than $6$ (even if we uniformly randomly choose $0.7$ fraction of the nodes to be malicious). Hence we assume $d = 6$ in our experiments.

We consider $f$ ranging from $0.4$ to $0.7$. Since \codename focuses on malicious majority, we do not consider smaller $f$ values. Recall that larger $s$ (i.e., number of fragments) gives higher throughput but longer confirmation latency. To strike a balance, our experiments always use $s=800$. To achieve an error probability $\epsilon \le 2^{-30}$, and following  Theorem~\ref{the:extension}, we use $m = 35$, $45$, $55$, and $80$, for $f = 0.4$, $0.5$, $0.6$, and $0.7$, respectively. 
We use a block size of $2$MB in \codename, and an inter-block time of roughly $68$, $74$, $81$,  and $98$ seconds for $f = 0.4$, $0.5$, $0.6$, and $0.7$, respectively. These parameters are chosen such that based on the analysis in Appendix~\ref{app:throughputanalysis}, each node consumes no more than about 90\% of its $20$Mbps available bandwidth (even in the very worst-case). Note that here we use the exact version of the analysis in Appendix~\ref{app:throughputanalysis}, without applying any approximation such as $\mathbb{Y} \approx \frac{wl}{s}$.

\begin{figure}[]
	\centering
	\includegraphics[width=\linewidth]{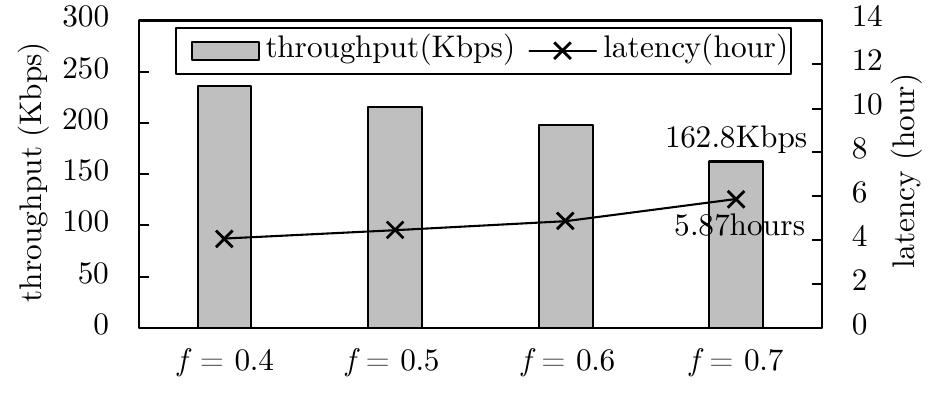}
\caption{End-to-end performance of \codename.}
	\label{fig:endtoendperformance}
	\vspace*{-5mm}
\end{figure}

\Paragraph{End-to-end performance.}
Figure~\ref{fig:endtoendperformance} plots \codename's transaction throughput and confirmation latency. As expected, the confirmation latency increases with $f$, since larger $f$ entails a larger committee size ($m$) and in turn more rounds in \tool. Similarly, the transaction throughput decreases with larger $f$ since as each invocation of \tool takes longer to finish, we need to correspondingly increase the inter-block time. This then decreases throughput. Nevertheless, even when $f=0.7$, \codename still achieves a throughput of about 163Kbps and a confirmation latency of less than $6$ hours. As explained in Section~\ref{sec:intro} where we used Bitcoin as a reference point, such performance is already  ``practically usable'': Bitcoin entails a confirmation latency of about $9.3$ hours to achieve $\epsilon \le 2^{-30}$ under $f=0.25$, and Bitcoin's throughput is about 14Kbps.

\begin{figure}[]
	\centering
	\includegraphics[width=\linewidth]{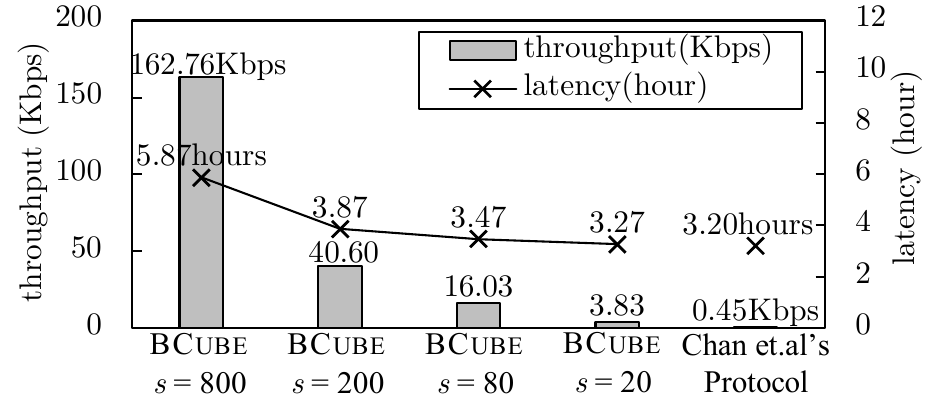}
	\vspace*{-2mm}
	\caption{Comparison of our protocol \codename and the state-of-the-art Chan et al.'s protocol~\cite{chan20}.}
	\label{fig_chan_comparison}
	\vspace*{-5mm}
\end{figure}

\Paragraph{Compare with state-of-the-art design.} 
There has been rather limited amount of prior work on designing blockchains for tolerating $f\ge 0.5$.
The current state-of-the-art approach is via Chan et al.'s protocol~\cite{chan20}.
Strictly speaking, Chan et al.'s protocol is a byzantine broadcast protocol, instead of a blockchain. But one could use Chan et al.'s protocol as the core to build a blockchain, in the same way as we use \tool to build \codename. To enable a direct comparison, we take \codename, and then replace \tool with Chan et al.'s protocol (as described in Section~\ref{sec:background2}), to obtain a blockchain based on their protocol. We use our own implementation of their protocol, since their work does not have implementation.
Section~\ref{sec:design1} explained that when running on multi-hop topologies, Chan et al.'s protocol would require infinite bandwidth if a malicious broadcaster keeps injecting conflicting messages. Our experiments for their protocol explicitly assume away this particular adversarial behavior --- doing so only makes the results for their protocol better.

Due to space constraints, we only present our results on Chan et al.'s protocol for $f=0.7$. Specifically, we measure the throughput/latency of Chan et al.'s protocol (i.e., the resulting blockchain), under the same setting as our \codename, such as 12-second round duration, same topology, around 98-second inter-block time, and a committee size of 80. We also use pipelined invocations for Chan et al.'s protocol, in the same way as we do in our protocol.
We set the block size to be such that in all cases (including all adversarial strategies and randomness), the bandwidth consumed by each node is no more than 90\% of the 20Mbps available bandwidth, which is the same constraint we imposed in the experiments of our \codename. Under such a constraint, the maximum block size we can use in the experiments for Chan et al.'s protocol is about $5.5$KB. 

Figure~\ref{fig_chan_comparison} compares the performance of our protocol and Chan et al.'s protocol, as observed in our experiments. Chan et al.'s protocol achieves a throughput of about $0.45$Kbps, while \codename achieves about $163$Kbps, which is over 350 times higher. 
Such large improvement primarily comes from the fact that \codename/\tool breaks each block into $s-1=799$ fragments and can delay the forwarding of individual fragments whenever needed, to avoid congestion in any given round. With some over-simplification, Chan et al.'s protocol can be viewed as having only a single fragment. This is also why their protocol can only use $5.5$KB block size, while we can support $2$MB block size. 

Using many fragments in \codename does increase the latency: 
Chan et al.'s protocol has a latency of $3.20$ hours, while ours is $5.87$ hours. To gain more insights, Figure~\ref{fig_chan_comparison} further presents the performance of \codename when using fewer fragments, with $s=200$, $s=80$, and $s=20$. In particular, with $s=20$ fragments, our latency is 3.27 hours, which is only $2.2$\% larger than their latency. Yet with $s=20$, we still achieve more than $850$\% of the throughput of their protocol, and can support block size of about $47$KB. Hence even if \codename is forced to provide almost the same latency as Chan et al.'s protocol, \codename still provides significantly higher throughput.

\section{Related Works}
\label{sec:related}

\Paragraph{Byzantine broadcast.}
Being a classic distributed computing problem, byzantine broadcast has been extensively studied. We will only focus on byzantine broadcast protocols~\cite{chan20,dolev83,ganesh16,hirt14,nayak20,tsimos20,wan20} that can tolerate $f \ge \frac{1}{2}$. Most of these are actually theoretical designs without implementation. 
Section~\ref{sec:intro} and \ref{sec:background2} already discussed \cite{chan20,dolev83,tsimos20}.
The protocols from \cite{hirt14,ganesh16,nayak20} all require direct point-to-point communication on a clique, and hence does not work for multi-hop topologies. Furthermore, these protocols are designed for a permissioned setting with a fixed set of $n$ nodes. The following nevertheless still reviews the techniques used in \cite{hirt14,ganesh16,nayak20}, and characterizes their TTB ratios.

Hirt and Raykov's protocol~\cite{hirt14} breaks the object into $n$ fragments, each with $\frac{l}{n}$ size, to optimize for communication complexity. For each fragment and each node, they invoke a smaller black-box byzantine broadcast protocol, resulting in total $n^2$ sequential invocations. 
Doing so enables later invocations to benefit from information collected during earlier invocations. 
The protocol takes total at least $n^2$ rounds. In some rounds, 
a node needs to send one fragment (i.e., $\frac{l}{n}$ bits). 
Hence the maximum $l$ the protocol can support, given $\mathbb{B}$ available bandwidth, is 
$l_0 = \mathbb{B}\delta n$. 
We thus have $\mathbb{T} \le \frac{l_0}{n^2 \delta} = \frac{\mathbb{B}\delta n}{n^2 \delta} = 
\frac{\mathbb{B}}{n}$ and $\mathbb{R} = \mathbb{T}/\mathbb{B} \le \frac{1}{n}$.
Ganesh and Patra's protocol~\cite{ganesh16} improves upon \cite{hirt14}, and reduces the time complexity to about $n$ rounds. In \cite{ganesh16}, some rounds are used for propagating the fragments. In each such round, a node may need to send its fragment to up to $n$ nodes, incurring $n\times \frac{l}{n} = l$ bits of communication. 
Hence the maximum $l$ the protocol can support, given $\mathbb{B}$ available bandwidth, is 
$l_0 = \mathbb{B}\delta$. 
We thus have $\mathbb{T} \le \frac{l_0}{n \delta} =
\frac{\mathbb{B}}{n}$ and $\mathbb{R} = \mathbb{T}/\mathbb{B} \le \frac{1}{n}$.
In comparison to \cite{ganesh16,hirt14}, \tool also breaks an object into fragments, but for a different purpose of improving throughput. Because of this, most issues in \tool such as delaying and compensation are not relevant to \cite{ganesh16,hirt14}.

Nayak et al.'s protocol~\cite{nayak20} further improves the communication complexity of  \cite{ganesh16}.
In their protocol, instead of sending the object to all other nodes directly, a node uses erasure coding and sends one fragment to each of the $n$ nodes. 
The $n$ nodes will then each forward its received fragment to all other nodes. Each fragment has size at least $\frac{l}{n}$, and hence a node needs to send at least $n\cdot \frac{l}{n} = l$ bits in some rounds. The maximum $l$ the protocol can support is then $l_0 = \mathbb{B}\delta$.
Their protocol has total $fn+1$ rounds. This leads to 
$\mathbb{T} \le \frac{l_0}{(fn+1) \delta}$ and $\mathbb{R} = \mathbb{T}/\mathbb{B} < \frac{1}{fn}$.
Their idea of using erasure coding is largely orthogonal to the techniques in \tool.

Finally, Wan et al.~\cite{wan20} recently propose a constant-round byzantine broadcast protocol for tolerating $f \ge \frac{1}{2}$.
When adapted to our multi-hop setting, their protocol takes at least $d$ rounds and in each round, a node may need to send the $l$-bit object to its $w$ neighbors. 
Hence the maximum $l$ the protocol can support is 
$l_0 = \frac{\mathbb{B}\delta}{w}$. In turn, $\mathbb{R} = \mathbb{T}/\mathbb{B} = \frac{l_0}{d \delta}/\mathbb{B} = \Theta(\frac{1}{d w})$.
In comparison, our \tool has $\mathbb{R} = \Theta(\frac{1}{w})$. More importantly, their protocol further needs each node to send up to $n^2$ bits or more (for additional protocol information) to each of its $w$ neighbors. Hence their protocol achieves $\mathbb{R} = \Theta(\frac{1}{dw})$ only when $l$ reaches the order of $wn^2$, which translates to about $500$MB under our experimental parameters. The block size in blockchains is typically much smaller than $500$MB. 

\Paragraph{Blockchains.}
Most existing blockchains (e.g., \cite{solida,prism19,ohie20,snow-white,algorand,ouroboros,omniledger,elastico,rapidchain}) 
today can only tolerate $f < \frac{1}{2}$. By leveraging the ``reputation'' of the nodes, RepuCoin~\cite{jiangshan19} can tolerate temporary malicious majority --- namely, temporary spikes in $f$ (but not $f\ge \frac{1}{2}$ in general).
While blockchains can be built from byzantine broadcast, and hence tolerate $f\ge \frac{1}{2}$, this fact has been largely neglected in the literature. 
Pass and Shi~\cite{pass2018thunderella} mention the design of a blockchain based on the Dolev-Strong protocol~\cite{dolev83} for byzantine broadcast. 
As explained in Section~\ref{sec:intro}, using the 
Dolev-Strong protocol
will result in rather low throughput (e.g., $0.072$Kbps). Our contribution is exactly to overcome this central issue. Note that the main focus of \cite{pass2018thunderella} is not on tolerating $f\ge \frac{1}{2}$, but on providing fast transaction confirmation when a super majority of the users are honest. In addition, their work is mainly theoretical, with no implementation. 

\section{Conclusions}
We have presented \codename, the very first blockchain that can tolerate $f\ge \frac{1}{2}$, while achieving practically usable
transaction throughput and latency. At the core of \codename is our novel byzantine broadcast protocol \tool, which can achieve significantly better throughput than prior protocols. 
\codename still leaves many questions unanswered. For example, can we further improve its performance? 
Can we generalize beyond Proof-of-Stake? 
Can we offer progressive confirmation, as in Bitcoin, so that a transaction's likelihood of being confirmed grows with time, even before it is fully-confirmed? 
All these are interesting open questions for future research.

\section*{Acknowledgment}
We thank the anonymous IEEE Security \& Privacy reviewers for their detailed and helpful comments on this paper.

\section*{Disclosure by Authors}
Haifeng Yu is an Associate Professor in School of Computing, National University of Singapore (NUS). Haifeng is also a Co-PI of NUS CRYSTAL Centre, which is a blockchain-related research centre. Prateek Saxena is an Associate Professor in School of Computing, NUS. Prateek is also a Co-Director of NUS CRYSTAL Centre, and a co-founder of Zilliqa Research, which is related to blockchains. 


\appendices

\renewcommand{\thesectiondis}[2]{\Roman{section}:}

\vspace*{1mm}
\section{Proof for Theorem~\ref{the:extension}}
\label{app:theextension}

\theextension*
\begin{proof}
	All probabilities in this proof, unless otherwise mentioned, are conditioned upon $\mathcal{Z}_\lambda(j-1)$.
	We only prove the harder case where $j\ge 2$ and $\lceil \frac{j}{\rho}\rceil \ne \lceil \frac{j+1}{\rho}\rceil$. (In other cases, the second part in $\mathcal{Z}_\lambda(j)$ trivially follows from the second part in $\mathcal{Z}_\lambda(j-1)$, and hence the proof is similar but easier.)
	Let $i = \lceil\frac{j+1}{\rho} \rceil -1$. We define several random events:
	\begin{itemize}
		\item $\mathcal{W}_1$: The committee for slot $j$ contains at least one honest member.
		\item $\mathcal{W}_2$: Among the first $\tau$ slots of epoch $i$, where each slot has a corresponding committee and broadcaster, there exists at least one slot whose broadcaster is honest. 
		\item $\mathcal{W}_3$: From time $T_1$ through $T_5$ in epoch $i$, the honest nodes and the adversary combined find no more than $\lambda$ PoW solutions. (This does not include PoW solutions found by the adversary prior to $T_1$, for example, when the PoW challenge is not fresh.)
	\end{itemize}
	We will later prove that:
	\begin{eqnarray}
	\Pr[\mathcal{W}_1] &\ge& 1 - \frac{\lambda f^m}{0.9\times (0.86-f^\tau)} \\
	\Pr[\mathcal{W}_2] &\ge& 1 - \frac{\lambda f^\tau}{0.9\times (0.86-f^\tau)} \\
	\Pr[\mathcal{W}_3] &\ge& 1-  {\tt Pois}(807, \lambda)
	\end{eqnarray}
	Hence with probability at least $1-\frac{\lambda f^m}{0.9(0.86-f^\tau)} - \frac{\lambda f^\tau}{0.9(0.86-f^\tau)} - {\tt Pois}(807, \lambda)$, all three events occur. Recall that the PoW challenge in epoch $i$ is the concatenation of all the nonces in the first $\tau$ slots. By events $\mathcal{Z}_\lambda(j-1)$, $\mathcal{W}_1$, $\mathcal{W}_2$, and Theorem~\ref{thm_agreement}, we have that i) all honest nodes agree on the PoW challenge in epoch $i$, and ii) the PoW challenge in epoch $i$ is fresh in the sense that the adversary does not see the challenge before $T_1$. Together with event $\mathcal{W}_3$, this means that no more than $\lambda$ PoW solutions are seen by honest nodes in epoch $i$, which we define as event $\mathcal{W}_4$. Finally, $\mathcal{Z}_\lambda(j)$ follows directly from $\mathcal{W}_1$, $\mathcal{W}_4$, and $\mathcal{Z}_\lambda(j-1)$.
	
	In the following, we analyze $\Pr[\mathcal{W}_1]$, $\Pr[\mathcal{W}_2]$, and $\Pr[\mathcal{W}_3]$. We start with $\Pr[\mathcal{W}_1]$.
	Consider any given $i'$ where $1\le i'\le i-1$, and PoW solution $x$ seen by some honest node in epoch $i'$. Let $y = {\tt hash}_2({\tt challenge}|x)$, where ${\tt challenge}$ is the PoW challenge corresponding to $x$. Essentially, $y$ is a potential beacon value for epoch $i'+1$, and may potentially further be reused later in epoch $i$. 
	If we choose the committee by using ${\tt hash}_1(\mbox{slot number}|y)$ as randomness, then the probability of the committee containing no honest member is at most $f^m$. 

Next, we upper bound the probability that $y$ is used/reused as the beacon in epoch $i$. Define $Z_1 = \mathcal{Z}_\lambda(i'\rho)$ and $Z_2 = \mathcal{Z}_\lambda(j-1)$. 
Conditioned upon $Z_1$ only, 
define $A_{i'}$ to be the random event where for every $g\in [i'+1, i-1]$, epoch $g$ satisfies at least one of the following two conditions: i) no honest node finds any PoW solution in epoch $g$, or ii) if $y$ were used as the beacon in epoch $g$, then none of the $\tau$ broadcasters in the $\tau$ candidate-holding blocks in epoch $g$ would be honest. Note that $A_{i'}$ is well-defined, even if  $Z_2$ does not happen, and even if the honest nodes do not agree on the PoW challenge in epoch $g$: In those cases, the first condition in $A_{i'}$ simply means that no honest node solves the PoW, based on whatever each honest node individually believes to be the PoW challenge. Define $p(i')$ to be the probability of $A_{i'}$ happening, conditioned upon $Z_1$ only. Define $q(i')$ to be the probability of $A_{i'}$ happening, conditioned upon $Z_2$. With a Poisson approximation and since the honest nodes on expectation find two PoW solutions in each epoch, we have $p(i') \le (f^\tau + 0.14)^{i-i'-1}$. In turn, by Bayes' formula and since $Z_2$ implies $Z_1$, we have 
$q(i') =  \frac{\Pr[A_{i'} Z_2]}{\Pr[Z_2]}  = \frac{\Pr[A_{i'} Z_1 Z_2]}{\Pr[Z_2]} \le \frac{\Pr[A_{i'} Z_1]}{\Pr[Z_2]} = 
\frac{\Pr[A_{i'}|Z_1]\Pr[Z_1]}{\Pr[Z_2]} \le \frac{\Pr[A_{i'}|Z_1]}{\Pr[Z_2]} = 
 \frac{p(i')}{\Pr[\mathcal{Z}_\lambda(j-1)]} \le \frac{(f^\tau + 0.14)^{i-i'-1}}{0.9}$.
Now conditioned upon $\mathcal{Z}_\lambda(j-1)$, in order for $y$ to be used as the beacon in epoch $i$, the event $A_{i'}$ must happen. Hence the probability of $y$ being used as the beacon in epoch $i$ is at most $\frac{(f^\tau + 0.14)^{i-i'-1}}{0.9}$. 

The probability of $y$ being used as the beacon in epoch $i$ and further causing the committee for slot $j$ to not contain any honest member is then at most $f^{m}\times \frac{(f^\tau + 0.14)^{i-i'-1}}{0.9}$.
Finally, there are at most $\lambda$ different $y$ values in each epoch $i'\in [1, i-1]$, and we need to take a union bound over all those. Hence we have $\Pr[\mathcal{W}_1]\ge 1- \sum_{i'=1}^{i-1}(\lambda f^{m} \frac{(f^\tau + 0.14)^{i-i'-1}}{0.9}) \ge 1 - \frac{\lambda f^m}{0.9(0.86-f^\tau)}$. 	
	
	We move on to $\Pr[\mathcal{W}_2]$. Each of the first $\tau$ slots in epoch $i$ has a corresponding broadcaster. $\mathcal{W}_2$ essentially is the event that at least one of these $\tau$ broadcasters is honest. Following similar reasoning as above, we have $\Pr[\mathcal{W}_2] \ge 1- \sum_{i'=1}^{i-1}(\lambda f^{\tau}\frac{(f^\tau + 0.14)^{i-i'-1}}{0.9})\ge 1 - \frac{\lambda f^\tau}{0.9(0.86-f^\tau)}$. 
	
	Finally we consider $\Pr[\mathcal{W}_3]$. Since the adversary has at most 100 times the computational power as honest nodes, and since the honest nodes on expectation find $2$ PoW solutions from $T_2$ to $T_3$, one can verify that on expectation the adversary finds no more than $805$ solutions from $T_1$ to $T_5$. Hence $\Pr[\mathcal{W}_3] \ge 
	1- {\tt Pois}(2+ 805, \lambda) = 1- {\tt Pois}(807, \lambda)$.
\end{proof}

\Paragraph{Remark.}
Theorem~\ref{the:extension} assumes that the adversary cannot adaptively corrupt honest nodes. If adaptive corruption is possible, then the adversary may corrupt the committee members after seeing the beacon and before the committee has done its work. Assume that the adversary takes at least $x$ epochs to adaptively corrupt nodes. Then 
for a given slot in epoch $i$ and following a similar reasoning as in the proof of Theorem~\ref{the:extension}, the probability of the adversary adaptively corrupting the committee members in that slot is at most $\sum_{i'=1}^{i-x+1}(\lambda \frac{(f^\tau + 0.14)^{i-i'-1}}{0.9}) \le \frac{\lambda(f^\tau + 0.14)^{x-2}}{0.9(0.86-f^\tau)} \approx \lambda\cdot 0.14^{x-2}$, which drops exponentially with $x$. If needed, the constant $0.14$ can be further decreased as well, by setting the PoW easier (and increasing $m$, $\tau$, and $\lambda$ accordingly).

\vspace*{1mm}
\section{Throughput Analysis of \codename}
\label{app:throughputanalysis}

This section analyzes the throughput of \codename. The throughput of \codename follows from the throughput of all the \tool invocations. For any given \tool invocation, define $\mathbb{Y}$ to be the maximum number of bits that an honest node needs to send in a round, with the maximum taken across all honest nodes, all rounds, all possible strategies of the adversary, and all possible randomness. Intuitively, $\mathbb{Y}$ is the very worst-case number of bits a node needs to send in a round. 
We derive $\mathbb{Y}$ first, and then derive throughput.

From the pseudo-code of Algorithm~\ref{alg_top} through 3, one can easily see that in each round, an honest node only sends messages at Line~\ref{line_root_push}, \ref{line_send_xi}, and \ref{line_frag_send}, {\em regardless of the attack strategy of the adversary and regardless of the randomness}. Furthermore, all these messages are always of fixed size.
Let $l_{\tt nonce}$, $l_{\tt hash}$, and $l_{\tt sig}$ be the size of a nonce, a hash, and an aggregate signature in Algorithm~\ref{alg_top}, respectively. Also recall that each message is sent to all the neighbors of the node, and each node has at most $w$ neighbors. Straight-forward counting shows that 
the total number of bits sent by each honest node in each round in one \tool invocation is at most $\mathbb{Y} = w\times (2\times (l_{\tt hash} + l_{\tt sig} + m)+
\max(\lceil\frac{l}{s-1}\rceil+ (l_{\tt hash} + 1)\cdot \lceil\log_2 s\rceil,\,\,
l_{\tt nonce} + (l_{\tt hash}+1)\cdot \lceil\log_2 s\rceil+ l_{\tt sig} + m))$, since:
\begin{itemize}
	\item At Line~\ref{line_root_push}, the total size of $r$ and $\sigma(r)$ is at most $l_{\tt hash} + l_{\tt sig} + m$.
	\item At Line~\ref{line_send_xi}, the size of $x_i$ (including the Merkle proof and the index $i$) is $\lceil\frac{l}{s-1}\rceil+ (l_{\tt hash} + 1)\cdot \lceil\log_2 s\rceil$.
	\item At Line~\ref{line_frag_send}, the total size of $x_s$ and $\sigma(x_s)$ is at most $l_{\tt nonce} + (l_{\tt hash}+1)\cdot \lceil\log_2 s\rceil+ l_{\tt sig} + m$.
\end{itemize} 
Under practical settings, including our experimental settings later, the term  $\lceil\frac{l}{s-1}\rceil$ is usually significantly (e.g., 10 times) larger than the other terms, and the value $s$ is usually not too small (e.g., $\ge 20$). In such cases, we simply have $\mathbb{Y} \approx \frac{wl}{s}$.

We now derive throughput.
Let $\mathbb{B}$ be the available bandwidth on each node. Recall that $\delta$ is the round duration, and let $\gamma$ denote the number of pipelined invocations 
of \tool that each node has at any point of time. 
Let $l = l_0$ be the solution for the equation $\mathbb{Y}\times \gamma = \mathbb{B} \times \delta$ (namely, the bandwidth needed in each round equals the bandwidth available). 
Each \tool invocation can thus confirm a block of size $l_0$ every $(2dm+s)\delta$ time.
The total throughput is then $\mathbb{T} = \frac{\gamma l_0}{(2dm+s)\delta}$. Under the approximation of $\mathbb{Y} \approx \frac{wl}{s}$, we have $l_0 \approx \frac{s\mathbb{B}\delta}{w\gamma}$ and $\mathbb{T} / \mathbb{B} \approx \frac{\gamma s\mathbb{B}\delta}{(2dm+s)\delta w \gamma}/\mathbb{B} = \frac{s}{(2dm+s)w}$. Setting $s = 2dm$ gives $\mathbb{R} = \mathbb{T} / \mathbb{B} \approx \frac{1}{2w} = \Theta(\frac{1}{w})$ and $\mathbb{T} \approx \frac{\mathbb{B}}{2w}$. 

Note that the above final throughput and TTB ratio is {\em independent} of $\gamma$, since eventually the $\gamma$ term gets cancelled out. Hence our final results hold {\em regardless of} whether pipelined invocations are used (i.e., whether $\gamma \ge 2$ or $\gamma=1$). This is expected: Using multiple pipelined invocations implies that each invocation gets only a portion of the available bandwidth and hence can only broadcast smaller-sized objects. These two factors, multiple invocations and smaller objects, cancel out. This is also consistent with the intuition that one cannot boost throughput, simply by using pipelined invocations.

To summarize, our analysis in this section shows that regardless of the attack strategy of the adversary, the total throughput of all the \tool invocations in \codename is 
$\mathbb{T} \approx \frac{\mathbb{B}}{2w} = \Theta(\frac{\mathbb{B}}{w})$, under the approximation of $\mathbb{Y} \approx \frac{wl}{s}$. 

\vspace*{1mm}
\section{Number of Aggregate Signature Signing/Validation Operations and Secure Hash Computations}
\label{appendix:crypto}

Via a careful calculation, this section shows that under all settings in Section~\ref{sec:exp}, 
{\em regardless of the strategy of the adversary and regardless of what messages the malicious nodes may send to the honest nodes}, in every second an honest \codename node only needs to do:
\begin{itemize}
    \item Adding a signer to a aggregate signature: at most 55 times
    \item Aggregate signature verification (where the signature {\em passes} verification): at most 55 times
    \item Aggregate signature verification (where the signature {\em fails} verification): at most 42 times
    \item Merkle proof verification (where the proof {\em passes} verification): at most 19 times
    \item Merkle proof verification (where the proof {\em fails} verification): at most 42 times
\end{itemize}
Under all settings in Section~\ref{sec:exp},
each Merkle proof verification takes no more than 10 secure hashes. 
Hence in every second, each honest \codename node only needs to do at most $55+55+42 = 152$ aggregate signature signing/validation operations, and at most $(19+42)\times 10 = 610$ secure hash computations for Merkle proof verification. Note that all these numbers are {\em worst-case} numbers: The actual numbers can be even smaller, for example, when there is no active attack on \codename.

The following calculates the worst-case number of various operations. In any given round of \tool, by the pseudo-code in Section~\ref{sec:pseudocode}, a node adds a signer to an aggregate signature at most 3 times. Under all settings in Section~\ref{sec:exp}, we have no more than 217 pipelined invocations of \tool at any given point of time. Since each round has 12 seconds, this translates to at most $3\times 217/12 < 55$ times/second. 

For verifying aggregate signatures, we use {\em lazy verification}: A node {\em only} verifies the signature on an item when it is about to use that item, instead of immediately upon receiving that item from its neighbors. For example in each round, a node may receive many Merkle roots from all its neighbors, but only picks the top two Merkle roots with the largest number of weighted signers, and processes those. The node will then simply verify the signatures (including the number of signers) on those two Merkle roots. If the signature does not pass verification, the node will pick again, until it gets two Merkle roots with valid signatures. One can then confirm, based on the pseudo-code in Section~\ref{sec:pseudocode}, that in each round, a node does at most 3 aggregate signature verifications where the signature {\em passes} verification. Hence the rate of such verification is again at most 55 times/second. 
By similar reasoning, one can confirm, based on the pseudo-code, that a node does at most 19 Merkle proof verifications (where the proof {\em passes} verification) per second. 

Finally, whenever an aggregate signature or Merkle proof does not pass verification, the neighbor who sent the corresponding item must be malicious. Hence a node {\em blacklists such a neighbor, and discards all previous/future messages from that neighbor}. With this simple trick, since each node has at most 42 neighbors in all our experiments, a node does at most 42 aggregate signature verifications where the signature {\em fails} verification, and at most 42 Merkle proof verifications where the proof {\em fails} verification.
	
\vspace*{1mm}
\section{Sanity Check on the Bandwidth Consumption of \codename Node}
\label{app:sanity}

This section provides a sanity check on the bandwidth consumption of a \codename node in our experiments. Our goal is to verify that each \codename node indeed never uses more than $20$Mbps bandwidth. Note that this does not directly follow from the $10$Gbps aggregate available bandwidth across the $500$ \codename nodes on one physical machine, since the $10$Gbps may not be shared evenly.
To do this sanity check, we pick an arbitrary node with the maximum degree of $42$ (larger degree leads to more bandwidth consumption), and allocate a PC to run only that node. We then directly measure the total network traffic on the Ethernet interface of that PCs in every second, using the linux bandwidth monitoring tool ${\tt bmon}$.

Figure~\ref{fig_bws}(a)-(c) plot such measured bandwidth consumption under $f=0.7$, as a fraction of $20$Mbps. Results under other $f$ values are similar. As expected, this fraction never exceeds $1.0$, confirming that the node indeed never uses more than $20$Mbps bandwidth. The zig-zag pattern in Figure~\ref{fig_bws}(c) is also expected: Recall that each \tool invocation has $2dm+s = 1760$ rounds. When there is no active attack, a node only needs to send messages in the first $800$ rounds. Also recall that in these experiments, at any point of time, a node has many pipelined \tool invocations. Based on such parameters, Figure~\ref{fig_bws2} plots the computed number of invocations that need to send messages, in every 1-second window.  Figure~\ref{fig_bws2} shows a similar zig-zag pattern as in Figure~\ref{fig_bws}(c),
which explains such a pattern.

	\begin{figure}[t]
	\centering
	\begin{subfigure}{\linewidth}
		\centering
		\includegraphics[width=\linewidth]{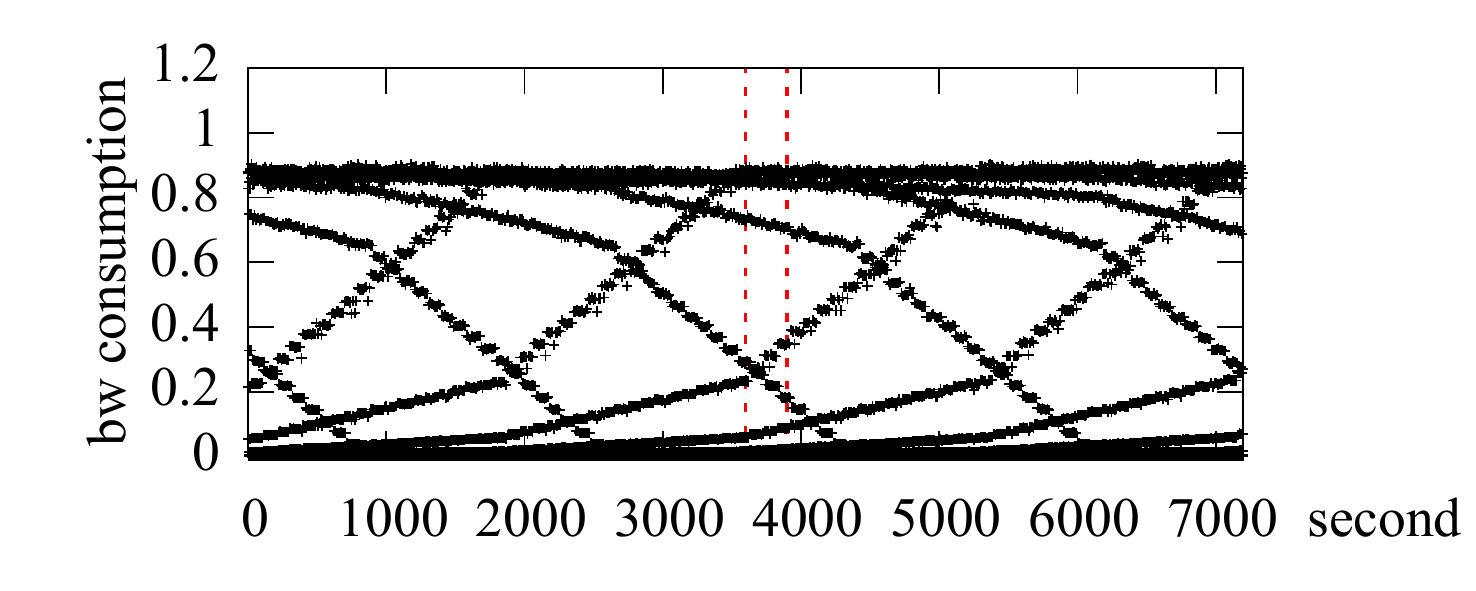}
		\vspace*{-8mm}
		\caption{bandwidth consumption in every second}
		\label{fig:big_bw}
	\end{subfigure}
	\begin{subfigure}{0.5\linewidth}
		\centering
		\includegraphics[width=\linewidth]{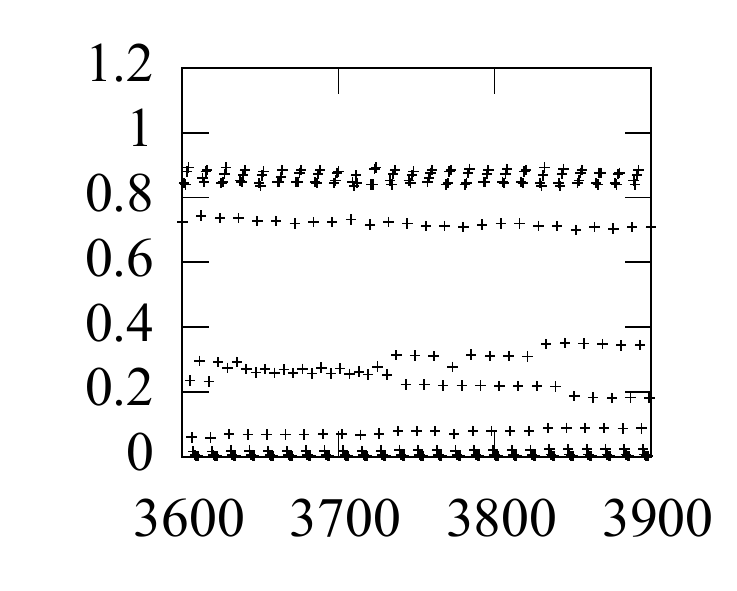}
		\vspace*{-8mm}
		\caption{zoomed in from above}
		\label{fig:medium_bw}
	\end{subfigure}\begin{subfigure}{0.5\linewidth}
		\centering
		\includegraphics[width=\linewidth]{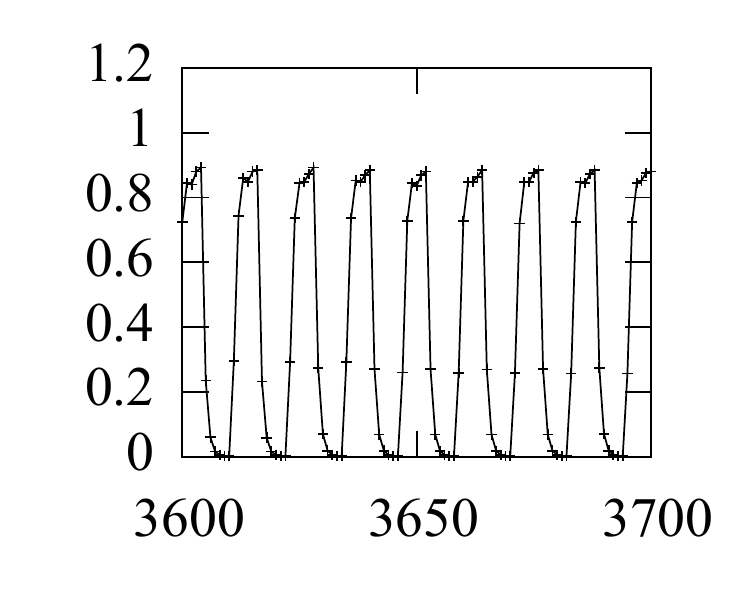}
		\vspace*{-8mm}
		\caption{zoomed in from above}
		\label{fig:small_bw}
	\end{subfigure}
	\vspace*{4mm}
	\caption{\label{fig_bws} Bandwidth consumption of an \codename node as a fraction of $20$Mbps. As expected, this fraction never exceeds $1.0$.
		Figure~\ref{fig_bws}(a) is plotted using points, but those dense points appear to be several curves. To make it clearer, we zoom into smaller time windows, using points in Figure~\ref{fig_bws}(b) and linespoints in Figure~\ref{fig_bws}(c).}
	\vspace*{0mm}
\end{figure}

\begin{figure}[t]
	\centering
	\includegraphics[width=0.65\linewidth]{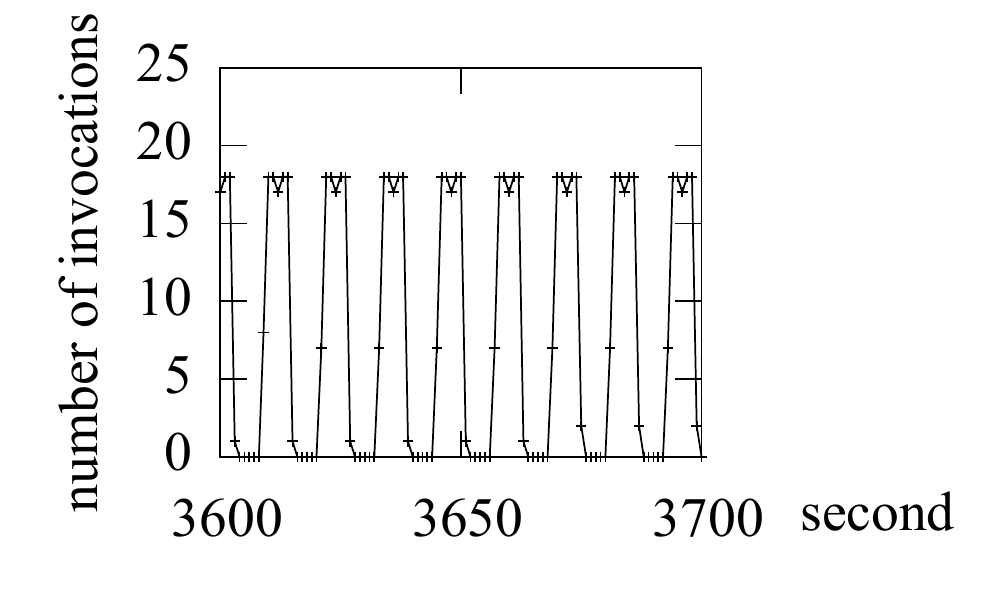}
	\caption{\label{fig_bws2} 
		Computed number of invocations that send messages.
	}
	\vspace*{-2mm}
\end{figure}
 
\vspace*{1mm}
\section{Proof for Theorem~\ref{thm_agreement}}
\label{app:theagreement}

This section proves Theorem~\ref{thm_agreement}. All line numbers in this section refer to lines in Algorithm 1 through 3.
We say that a node {\em accepts} a Merkle root $r$ if the node adds $r$ to its ${\tt root\_accepted}$ at either Line~\ref{line_committee_accept_root} or Line~\ref{line_noncommittee_accept_root}. A node may accept the same $r$ multiple times. In the overlay network, we call a path as an {\em honest path} if it (including the starting and ending node) contains only honest nodes and good edges. 
The {\em honest distance} between two honest nodes $A$ and $B$ is the length of the shortest honest path between $A$ and $B$. The proofs will use superscript to indicate variables on a give node --- for example, ${\tt top\_root}^C$ refers to the ${\tt top\_root}$ in the algorithm running on node $C$.

\Paragraph{Roadmap.}
The following is a roadmap for the proofs. Appendix~\ref{app:rootagree} 
presents several lemmas and then Theorem~\ref{thm_root}, which shows that the first phase enables the honest nodes to agree on the Merkle root. Appendix~\ref{app:fragagree} eventually gives Theorem~\ref{thm_frag}, which captures the agreement property of the second phase (for the fragments). Finally, Appendix~\ref{app:finalproof} proves Theorem~\ref{thm_agreement}, by using Theorem~\ref{thm_root} and \ref{thm_frag}.

\vspace*{1mm}
\subsection{Agreement on the Merkle Root}
\vspace*{1mm}
\label{app:rootagree}

\ifthenelse{\boolean{short}}{
	The proofs of all the lemmas/theorem in this section are deferred to \cite{badmajority_techreport}.
}{}
Lemma~\ref{lem_comm_acc} next roughly says that if an honest committee member $A$ accepts a certain Merkle root $r_0$, then all other honest nodes must also accept $r_0$ within some rounds after that, assuming the algorithm has not already terminated by then. But 
there will be an exception --- an honest node may accept two different roots $r_1$ and $r_2$, without accepting $r_0$.

\begin{lemma}\label{lem_comm_acc}
	Consider any honest committee member $A$ and any honest node $D$, and let $g\in[0,d]$ be the honest distance between $A$ and $D$. If $A$ accepts $r_0$ in round $i$ and if $i+g\le 2dm+s-1$, then by round $i+g$, node $D$ must satisfy either one or both of the following properties:
	\begin{itemize}
		\item $D$ accepts $r_0$.
		\item $D$ accepts two different roots.
	\end{itemize} 
\end{lemma}
\ifthenelse{\boolean{short}}{
}{
	\begin{proof}
		Obviously, the case for $g=0$ (implying $A=D$) is trivial, and we only prove for $g\in[1,d]$. 
		Let $i_1\leq i$ be the first round when $A$ accepts $r_0$. Since $A$ first accepts $r_0$ in round $i_1$, $A$ must satisfy the condition of $2d|\sigma^A(r_0)| \geq i_1$ at Line~\ref{line_root_comm} in that round. $A$ must then immediately adds its own signature to $\sigma^A(r_0)$. Then at Line~\ref{line_root_push}, $A$ must send to its neighbors $r_0$ together with an aggregate signature containing at least $\lceil\frac{i_1}{2d}\rceil+1$ weighted signers for $r_0$.
		
		The remainder of the proof relies on the following claim, which will be later proved: For any honest node $C$ with $g_1\in [1,g]$ honest distance from $A$, at least one of the following properties must hold in round $i_1+g_1$ immediately before Line \ref{line_root_pick} on node $C$:
		\begin{itemize}
			\item $r_0\in {\tt top\_root}^C$ and $|\sigma^C (r_0)| \geq \lceil \frac{i_1}{2d} \rceil+1$
			\item $|{\tt top\_root}^C| =2$ and $|\sigma^C (r)| \geq \lceil \frac{i_1}{2d} \rceil+1$ for all $r\in {\tt top\_root}^C$
		\end{itemize}
		Applying the above claim to node $D$ with $g_1=g$ then shows that in round $i_1+g\leq i+g$, node $D$ must satisfy the condition at either Line \ref{line_root_comm} (if $D$ is a committee member) or Line \ref{line_root_noncomm} (if $D$ is a non-committee member). This implies that by round $i+g$, $D$ must either accept $r_0$ or accept two different roots (in ${\tt top\_root}^D$), which gives us the lemma.
		
		The following proves the earlier claim via an induction. 
		For $g_1=1$, recall that $A$ has sent $r_0$ with a signature containing at least $\lceil\frac{i_1}{2d}\rceil+1$ weighted signers in round $i_1$. Node $C$ must receive these by the beginning of round $i_1+1$. Next in round $i_1+1$, if $r_0 \in {\tt top\_root}^C$ on node $C$, we are done. Otherwise we must have $r_0 \notin {\tt top\_root}^C$ and $r_0 \in {\tt all\_root}^C$, which means ${\tt top\_root}^C\ne {\tt all\_root}^C$. By Line~\ref{line_root_1root}, this implies that $|{\tt top\_root}^C| = 2$. Furthermore, since $r_0$ is not among the two Merkle roots with an aggregate signature containing most weighted signers, we have $|\sigma^C(r)| \geq \lceil\frac{i_1}{2d}\rceil+1$ for all $r \in {\tt top\_root}^C$, and we are done as well.
		
		Next we assume the claim holds for $g_1$, and we prove for $g_1+1$. Since $C$ has honest distance of  $g_1+1$ from $A$, let honest node $B$ be the second to the last node on the honest path from $A$ to $C$. By inductive hypothesis, the claim holds for $B$ in round $g_1$. 
		If we have $r_0\in {\tt top\_root}^B$ in round $i_1+g_1$, then in that round $B$ must send $C$ the root $r_0$ together with an aggregate signature on $r_0$ containing at least $\lceil\frac{i_1}{2d}\rceil+1$ weighted signers. Next in round $i_1+g_1+1$, if $r_0 \in {\tt top\_root}^C$ on node $C$, we are done. Otherwise by similar argument as earlier, we must have $|{\tt top\_root}^C| = 2$ and all roots in ${\tt top\_root}^C$ must have an aggregate signature containing at least $\lceil\frac{i_1}{2d}\rceil+1$ weighted signers, and we are done as well.
		
		If we have $|{\tt top\_root}^B| = 2$ in round $i_1+g_1$, with all roots in ${\tt top\_root}^B$ having an aggregate signature containing at least $\lceil\frac{i_1}{2d}\rceil+1$ signs, then in that round $B$ must send to $C$ the two roots in ${\tt top\_root}^B$, together with the signatures on them. Next on node $C$ immediately before Line~\ref{line_root_pick} in round $i_1+g_1+1$, we must have that $|{\tt all\_root}^C| \geq 2$, $|{\tt top\_root}^C| = 2$, and all roots in ${\tt top\_root}^C$ have an aggregate signature containing at least $\lceil \frac{i_1+g}{2d} \rceil$ weighted signers. This completes our inductive proof for the claim.
	\end{proof}
}

Lemma~\ref{lem_push_score_g} and \ref{lem_noncomm_acc}
next intend to eventually show that if an honest non-committee member $A$ accepts a certain Merkle root $r_0$, then all honest committee members must also accept $r_0$ within some rounds after that. Same as in Lemma~\ref{lem_comm_acc}, there will be an exception --- namely, accepting two different roots $r_1$ and $r_2$ instead of $r_0$. 

\begin{lemma}\label{lem_push_score_g}
	Consider any honest node $A$ and any honest committee member $D$, and let $g\in [0,d]$ be the honest distance between $A$ and $D$. If at Line~\ref{line_root_push} of round $i$, node $A$ makes a push with a score of at least $g$, then we must have $i+g\le 2dm+s-1$, and furthermore node $D$ must satisfy either one or both of the following properties in round $i+g$:
	\begin{itemize}
		\item $D$ accepts $r_0$, where $r_0$ is the root contained in $A$'s push. 
		\item $D$ accepts two different roots.
	\end{itemize} 
\end{lemma}
\ifthenelse{\boolean{short}}{
}{
	\begin{proof}
		We only prove the lemma for $g\in[1,d]$ --- the case for $g=0$ (implying $A=D$) can be proved in a similar (and easier) way. 
		Let $p$ be the push done by $A$ as specified in the lemma. We have $p.{\tt score} = 2d|\sigma^A(r_0)|-i$. Hence we have $i+g = 2d|\sigma^A(r_0)| -p.score + g \le 2dm - g+g \leq 2dm+s-1$. We also have $|\sigma^A(r_0)| = \frac{i+p.{\tt score}}{2d} = 
		\lceil \frac{i+p.{\tt score}}{2d}\rceil \geq \lceil \frac{i+g}{2d} \rceil$.
This means that in round $i$, node $A$ sends to its neighbors $r_0$ together with an aggregate signature containing at least $\lceil\frac{i+g}{2d}\rceil$ weighted signers on $r_0$. 
		
		To prove the rest of the lemma, we rely on the following claim: Let $i_1= i$. 
		For any honest node $C$ with
		$g_1\in[1, g]$ honest distance from $A$, at least one of the following must hold in round $i_1+g_1$ immediately before Line \ref{line_root_pick} on node $C$:
		\begin{itemize}
			\item $r_0\in {\tt top\_root}^C$ and $|\sigma^C (r_0)| \geq \lceil \frac{i_1+g}{2d} \rceil$
			\item $|{\tt top\_root}^C| =2$ and $|\sigma^C (r)| \geq \lceil \frac{i_1+g}{2d} \rceil$ for all $r\in {\tt top\_root}^C$
		\end{itemize}
		The above claim can be proved using exactly the same proof as in Lemma~\ref{lem_comm_acc}, after replacing ``$\lceil \frac{i_1}{2d} \rceil+1$'' with ``$\lceil \frac{i_1+g}{2d} \rceil$''. For clarity, we do not repeat the proof here.
		
		Applying the above claim to node $D$ with $g_1=g$ and $i_1 = i$ then shows that in round $i+g$, node $D$ must satisfy the condition at Line~\ref{line_root_comm}, and hence must either accept $r_0$ or accept two different roots (in ${\tt top\_root}^D$).
	\end{proof}
}

\begin{lemma}\label{lem_noncomm_acc}
	Consider any honest non-committee member $A$ and any honest committee member $D$, and let $g\in [1,d]$ be the honest distance between $A$ and $D$. If $A$ accepts $r_0$ in round $i$, then we must have $i+g\le 2dm+s-1$, and furthermore node $D$ must satisfy either one or both of the following properties in round $i+g$:
	\begin{itemize}
		\item $D$ accepts $r_0$.
		\item $D$ accepts two different roots.
	\end{itemize} 
\end{lemma}
\ifthenelse{\boolean{short}}{
}{
	\begin{proof}
		Since $A$ accepts $r_0$ in round $i$, then at Line~\ref{line_root_noncomm} node $A$ must see $2d|\sigma^A(r_0)|\ge i+d$. Let $p$ be the push corresponding to the send done by $A$ in that round. Then $p.{\tt score} = 2d|\sigma^A(r_0)|- i\ge (i+d) -i = d\ge g$. Invoking Lemma~\ref{lem_push_score_g} then gives to the current lemma.
	\end{proof}
}

Intuitively, Theorem~\ref{thm_root} next implies that exactly one of the following cases must happen at Line~\ref{line_final}:
\begin{itemize}
	\item All honest nodes have the same singleton set as the value for ${\tt root\_accepted}$; or
	\item $|{\tt root\_accepted}|\ne 1$ on all honest nodes. (In this case, all honest nodes will eventually output $\bot$.)
\end{itemize}

\begin{theorem} [Agreement on Merkle Root] \label{thm_root}
	Consider any execution of Algorithm~\ref{alg_top}, where at least one honest node has $|{\tt root\_accepted}| = 1$ at Line~\ref{line_final}. Then in this execution, all honest nodes must have the same ${\tt root\_accepted}$ value at Line~\ref{line_final}.
\end{theorem}
\ifthenelse{\boolean{short}}{
}{
	\begin{proof}
		Let node $B$ be any honest node with 
		$|{\tt root\_accepted}^B| = 1$ at Line~\ref{line_final}. Let $\{r_1\} = {\tt root\_accepted}^B$. We prove via a contradiction, and assume that there exists some honest node $A$ with ${\tt root\_accepted}^A \ne {\tt root\_accepted}^B$. Let $X = {\tt root\_accepted}^A \setminus {\tt root\_accepted}^B$. 
		
		\Paragraph{If $X$ is not empty.} If $X \ne \emptyset$, let $r_0$ be any element of $X$. We must have $r_0\ne r_1$. 
		We now derive a contradiction from the existence of $r_0$, by considering two cases. The first case is where $A$ is a committee member. Let $i$ be the first round during which $A$ accepts $r_0$. At Line~\ref{line_root_comm} in that round, we must have $|\sigma^A(r_0)| \le m-1$, which implies that $i \le 2d |\sigma^A(r_0)| \le 2d(m-1)$. Let $g_1$ be the honest distance between $A$ and $B$. Since $i+g_1\le i+d \le 2d(m-1)+ d < 2dm+s-1$, Lemma~\ref{lem_comm_acc} tells us that by the end of the execution, node $B$ must either accept $r_0$ (which contradicts with $r_0\in {\tt root\_accepted}^A \setminus {\tt root\_accepted}^B$) or accept two different roots (which contradicts with $|{\tt root\_accepted}^B| = 1$).  
		
		The second case is where $A$ is a non-committee member. Let $i$ be the first round during which $A$ accepts $r_0$. Let $D$ be any honest committee member. Let $g_2$ be the honest distance between $A$ and $D$, and $g_3$ be the honest distance between $D$ and $B$. By Lemma~\ref{lem_noncomm_acc}, $D$ must either accept $r_0$ or accept two different roots in round $i+g_2$. In either case, among the root(s) that $D$ accepts in round $i+g_2$, there must exists some root $r_2$ ($r_2$ may or may not equal $r_0$) such that $r_2\ne r_1$. 
		
		Let round $j \le i+g_2$ be the first round during which $D$ accepts $r_2$. At Line~\ref{line_root_comm} in that round, we must have $|\sigma^D(r_2)| \le m-1$, which implies that $j \le 2d |\sigma^D(r_0)| \le 2d(m-1)$. Since $j+g_3\le j+d \le 2d(m-1)+ d < 2dm+s-1$, Lemma~\ref{lem_comm_acc} tells us that by the end of the execution, node $B$ must either accept $r_2$ or accept two different roots. In either case, this contradicts with ${\tt root\_accepted}^B = \{r_1\}$.
		
		\Paragraph{If $X$ is empty.} If $X = \emptyset$, then since ${\tt root\_accepted}^A \ne {\tt root\_accepted}^B$ and since ${\tt root\_accepted}^B = \{r_1 \}$, we must have 
		${\tt root\_accepted}^A = \emptyset$. Given that $B$ accepts $r_1$ in some round, using a similar proof as above, one can show that $A$ must accept at least one root. This then contradict with ${\tt root\_accepted}^A = \emptyset$.    
	\end{proof}
}

\vspace*{1mm}
\subsection{Agreement on Fragments}
\label{app:fragagree}
\vspace*{1mm}

\ifthenelse{\boolean{short}}{
	The proofs of all the lemmas/theorem in this section are deferred to \cite{badmajority_techreport}.
}{}
Recall from Section~\ref{sec:design4} that in the second phase, a node $B$ uses the Merkle root contained in its most promising push done so far in the first phase, as $B$'s current guess for the final accepted Merkle root. Lemma~\ref{lem_high_push} below says that under certain conditions, after an honest node $A$ accepts a Merkle root, within a certain number of rounds, the guesses made by other honest nodes will become correct.

\begin{lemma}
	\label{lem_high_push}
	Consider any given execution of Algorithm~\ref{alg_top}, where $|\tt root\_accepted| = 1$ at Line~\ref{line_final} on some honest node $A$. Let round $i$ be when $A$ first accepts the sole element $r_0$ in ${\tt root\_accepted}^A$. Let $B$ be any honest node ($B$ can be $A$ itself), and let $g\in[0,d]$ be the honest distance between $A$ and $B$. Then in round $i + g$ and all later rounds, the push $p$ (i.e., the most promising push) chosen by node $B$ at Line~\ref{line_frag_pick_root} must contain $r_0$, if either of the following two conditions is satisfied:
	\begin{itemize}
		\item $A$ is a committee member.
		\item $A$ is a non-committee member and there exists some honest committee member $D$ such that the honest distance between $B$ and $D$ is no more than $d-g$.
	\end{itemize}
\end{lemma}
\ifthenelse{\boolean{short}}{
}{
	\begin{proof}
		We prove this lemma via a contradiction --- assume that in some round $j\ge i+g$, the push (i.e., the most promising push) chosen by node $B$ at Line~\ref{line_frag_pick_root} does not contain $r_0$. Let this push be $p_1$, and let $r_1\ne r_0$ be the root contained in $p_1$. We consider the two cases as listed in the lemma.
		
		\Paragraph{First case:} $A$ is a committee member.  
		We will later prove that $p_1.{\tt score} \ge 2d-g$. 
		Let $j_1\le j$ be the round during which node $B$ did the push $p_1$. We claim that since $p_1.{\tt score} \ge 2d-g$, node $B$ must have accepted the root $r_1$ contained in $p_1$ in round $j_1$. To see why, observe that $p_1.{\tt score} \ge 2d-g$ implies that $2d|\sigma^B(r_1)| - j_1 \ge 2d-g\ge d$
		and $2d|\sigma^B(r_1)|\ge j_1 + d$ at Line~\ref{line_root_calculate} in round $j_1$. If $B$ 
		is a non-committee member, it would have previously satisfied the condition at  Line~\ref{line_root_noncomm} in round $j_1$. If $B$ is a committee member, it must have previously satisfied the condition at  Line~\ref{line_root_comm} in round $j_1$ --- otherwise we would have 
		$2d|\sigma^B(r_1)|< j_1$ at Line~\ref{line_root_calculate}, contradicting with $2d|\sigma^B(r_1)|\ge j_1 + d$. Hence regardless of whether $B$ is a committee member, $B$ must have accepted $r_1$. Since $r_1\ne r_0$ and since ${\tt root\_accepted}^A = \{r_0\}$, this contradicts with Theorem~\ref{thm_root}. 
		
		The following proves that $p_1.{\tt score} \ge 2d-g$. 
		Since $A$ first accepts $r_0$ in round $i$, we know that in Line~\ref{line_root_comm} of that round, $A$ must see $2d|\sigma^A(r_0)| \ge i$. Then $A$ will add its own signature for $r_0$, which means that $A$ must send to its neighbors $r_0$ together with an aggregate signature containing at least $\lceil \frac{i}{2d} \rceil +1$ weighted signers on $r_0$. The push $p_2$ corresponding to this send done by $A$ in round $i$ has a score of at least $2d(\lceil \frac{i}{2d} \rceil +1) - i \ge 2d$.
		
		If $g =0$, then $A=B$ and $B$ has made a push $p_2$ with score at least $2d$ in round $i$. Since $p_1$ is the push chosen by $B$ as the push with the highest score in round $j \ge i+g = i$, we must have $p_1.{\tt score} \ge p_2.{\tt score} \ge 2d \ge 2d-g$.
		
		If $g\ge 1$, 
		consider the $k$-th node $C$ on the honest path from $A$ (exclusive) to $B$ (inclusive). A trivial induction can show that in round $i+k$, node $C$ sends some root $r_2$ with an aggregate signature on $r_2$ containing $\lceil \frac{i}{2d} \rceil +1$ weighted signers. (Here $r_2$ may or may not equal $r_0$. Also, $C$ may send one more root beyond $r_2$.)
		This means that $B$ must send some root $r_3$ with a signature on $r_3$ containing at least $\lceil \frac{i}{2d} \rceil +1$ weighted signers in round $i+g$. Let $p_3$ be the push corresponding to this send done by $B$. We immediately know that $p_3.{\tt score} \ge 2d (\lceil \frac{i}{2d} \rceil +1) -(i+g) \ge 2d -g$. Since $p_1$ is chosen in round $j\ge i+g$ as the push with the highest score, we must have $p_1.{\tt score} \ge p_3.{\tt score} \ge 2d-g$. 
		
		\Paragraph{Second case:} $A$ is a non-committee member and there exists some honest committee member $D$ such that the honest distance between $B$ and $D$ is no more than $d-g$.
		
		In this case, in round $i$ node $A$ must send to its neighbors $r_0$ together with an aggregate signature on $r_0$ containing at least $\lceil \frac{i+d}{2d} \rceil$ weighted signers. In turn, by similar reasoning as above, we can show that $p_1.{\tt score} \ge 2d \lceil \frac{i+d}{2d} \rceil -(i+g) \ge d -g$. Since the distance between $B$ and $D$ is no more than $d-g$, Lemma~\ref{lem_push_score_g} tells us that $D$ must either accepts $r_1$ or accepts two different roots by the end of the execution. In either case, this contradicts with Theorem~\ref{thm_root} since ${\tt root\_accepted}^A = \{r_0\}$.
	\end{proof}
}

Lemma~\ref{lemma:frag1} and \ref{lemma:frag2} next reason about the agreement properties for the fragments, under certain conditions.

\begin{lemma}\label{lemma:frag1}
	Consider any given execution of Algorithm~\ref{alg_top}, where at least one honest node has $|{\tt root\_accepted}| = 1$ at Line~\ref{line_final}. If there exists some honest committee member having ${\tt frag\_accepted} = {\tt true}$ at Line~\ref{line_final}, then all honest nodes must have ${\tt frag\_accepted}$ being true at Line~\ref{line_final}.
\end{lemma}
\ifthenelse{\boolean{short}}{
}{
	\begin{proof}
		Let $A$ be the first honest committee member that sets its ${\tt frag\_accepted}^A$ to be true. 
		Consider any given honest node $C$ whose honest distance to $A$ is $g\in[1,d]$. We will show that $C$ must set ${\tt frag\_accepted}^C$ to be true by Line~\ref{line_final}.
		
		\Paragraph{Properties on $A$.}
		Let $j$ be the value of $t_{\textnormal{frag}}^A$ when $A$ first sets ${\tt frag\_accepted}^A$ to be true. We claim that $j+g\le 2dm+s-1$, namely, round $j+g$ must be before the end of the execution: 
		Since this is the first time that $A$ sets ${\tt frag\_accepted}^A$ to be true, at Line~\ref{line_frag_comm} node $A$ must have $|\sigma^A(x_s)| \leq m-1$ and $2d|\sigma^A(x_s)| \geq j-(s-1)$. Thus we have $j+g\le 2d|\sigma^A(x_s)|+ (s-1)+ g \leq 2d(m-1) + (s-1) + d < 2dm + s - 1$. 
		
		By Theorem~\ref{thm_root}, we know that all honest nodes must have the same singleton set ${\tt root\_accepted}$ at Line~\ref{line_final}. Let $r_0$ be the sole element in this set, and let $y$ be the last fragment corresponding to the Merkle root $r_0$. Let $x$, $z$, and $t_0$ be the value of $x^A_s$, $\sigma^A(y)$, and round number on $A$, respectively, when $A$ first sets ${\tt frag\_accepted}^A$ to be true. We next prove that $x = y$. 
		Since $j+g < 2dm + s - 1$, we must have $j\ne \infty$ and hence $t_{\textnormal{root}}^A\ne \infty$ at Line~\ref{line_frag_t2}. This means that $t_{\textnormal{root}}^A$ has already been assigned some value at Line~\ref{line_root_t0_comm} during or before round $t_0$. Since $t_{\textnormal{root}}^A$ is never assigned a value larger than the current round, we must have 
		$t_{\textnormal{root}}^A \le t$. Lemma~\ref{lem_high_push} tells us that starting from round $t_{\textnormal{root}}^A$ (inclusive), the most promising push at Line~\ref{line_frag_pick_root} on node $A$ must always contain $r_0$.
		Hence we have $x = y$.
		
		\Paragraph{Property for $C$.}
		We next prove the following claim via an induction on $g\in[0,d]$: 
		\begin{itemize}
			\item {\bf Claim 1}. By the end of round $j+g$, node $C$ must have forwarded/sent all $s$ fragments corresponding to the Merkle root $r_0$, as well as $z$, to its neighbors. 
		\end{itemize}
		For $g=0$, note that $x=y$, and hence in round $t_0$ node $A$ must send $y$ and $z$ to its neighbors. Since $t_0\le j$ and since $A$ can only send $y$ after it has previously sent the other $s-1$ fragments corresponding to $r_0$, the induction base holds. 
		
		Next, assume that  Claim 1 holds for $g$ and we consider $g+1$. Let node $B$ be the second to the last node on the shortest honest path from $A$ to $C$. By inductive hypothesis, by the end of round $j+g$, node $B$ must have forwarded/sent all $s$ fragments corresponding to $r_0$, as well as $z$, to $C$. Hence, $C$ must have received $y$ (together with $z$) from $B$ at the beginning of round $j+g+1$. 
		Lemma~\ref{lem_high_push} tells us that starting from round $t_{\textnormal{root}}^A+g+1$ (inclusive), the most promising push  at Line~\ref{line_frag_pick_root} on node $C$ must always contain $r_0$. From round $t_{\textnormal{root}}^A+g+1$ to round $j+g+1$ (both inclusive), there are at least total $s$ rounds. Since $C$ receives $y$ (together with $z$) from $B$ by the last round among these $s$ round, and since $B$ can only send one fragment to $C$ in each round, node $C$ must have sent all the $s$ fragments (together with $z$) by the end of round $j+g+1$. This completes the inductive proof for Claim 1. 
		
		\Paragraph{Tracing execution on $C$.}
		We next show that $C$ must set ${\tt frag\_accepted}^C$ to be true in round $j+g$, by tracing all the steps in Algorithm 3 in that round:
		\begin{itemize}
			\item {\bf Line~\ref{line_frag_nopush}:} Because $t_{\textnormal{root}}^A+g \le j+g \le 2dm+s-1$ (the second ``$\le$'' was proved earlier), Lemma~\ref{lem_comm_acc} tells us that node $C$ must have $t_{\textnormal{root}}^C \le t_{\textnormal{root}}^A+g\le j+g$. This means that by Line~\ref{line_frag_nopush} in round $j+g$, $C$ has already accepted $r_0$ --- hence ${\tt all\_push}^C \ne \emptyset$.
			
			\item {\bf Line~\ref{line_frag_pick_root}:} Since $C$ has already accepted $r_0$ in or before round $j+g$, Lemma~\ref{lem_high_push} tells us that the most promising push chosen by node $C$ at Line~\ref{line_frag_pick_root} must contain $r_0$.

			\item {\bf Line~\ref{line_frag_allfrag} and \ref{line_frag_receive_last}:} Claim 1 tells us that in round $j+g$, node $C$ will not satisfy the condition at Line~\ref{line_frag_allfrag}, and will satisfy the condition at Line~\ref{line_frag_receive_last}.
			
			\item{\bf Line~\ref{line_frag_comm}:} If $C$ is a committee member, we claim that the condition at Line~\ref{line_frag_comm} must be satisfied. Assume otherwise. First,
			recall that we showed earlier that $t_{\textnormal{root}}^C \le t_{\textnormal{root}}^A+g$. Hence $t_{\textnormal{frag}}^C = \max(j+g,t_{\textnormal{root}}^C+(s-1)) \le \max(j+g, t_{\textnormal{root}}^A+g+(s-1))= j+g$. Next, since node $A$ sets $\tt frag\_accepted$ to be true in round $t_0$, at Line~\ref{line_frag_comm} node $A$ must satisfy 
			$2d|\sigma^A(x_s)| \geq j - (s-1)$, at Line~\ref{line_frag_addsig} $A$ must add itself as a signer to the aggregate signature, and hence $2d|z|\ge j-(s-1)+2d$.
			Finally, by Claim 1, node $C$ has sent $z$ to all its neighbors by the end of round $j+g$. Hence at Line~\ref{line_frag_comm} on $C$, we must have $2d|\sigma^C(y)|\ge 2d|z|\ge j-(s-1)+2d \ge j-(s-1)+g \ge t_{\textnormal{frag}}^C -(s-1)$. This contradicts with the assumption that the condition at Line~\ref{line_frag_comm} is not satisfied.
			
			\item {\bf Line~\ref{line_frag_noncomm}:} If $C$ is a non-committee member, a similar proof as above can show that the condition at Line~\ref{line_frag_noncomm} must be satisfied.
			
			\item {\bf Line~\ref{line_frag_comm_set} and \ref{line_frag_noncomm_set}:} Since $C$ must satisfy either the condition at Line~\ref{line_frag_comm} or Line~\ref{line_frag_noncomm}, $C$ will reach either Line~\ref{line_frag_comm_set} or \ref{line_frag_noncomm_set}, and will set ${\tt frag\_accepted}$ in round $j+g$, which is before the end of the execution. 
		\end{itemize}
	\end{proof}
}

\begin{lemma}\label{lemma:frag2}
	Consider any given execution of Algorithm~\ref{alg_top}, where at least one honest node has $|{\tt root\_accepted}| = 1$ at Line~\ref{line_final}. If there exists some honest non-committee member having ${\tt frag\_accepted} = {\tt true}$ at Line~\ref{line_final}, then all honest committee members must have ${\tt frag\_accepted}$ being true at Line~\ref{line_final}.
\end{lemma}
\ifthenelse{\boolean{short}}{
}{
	\begin{proof}
		Let $A$ be any honest non-committee member having ${\tt frag\_accepted} = {\tt true}$ at Line~\ref{line_final}.
		Let $C$ be any honest committee member, and let $g\in[1,d]$ be the honest distance between $A$ and $C$. 
		
		\Paragraph{Properties on $A$.}
		Let $j$ be the value of $t_{\textnormal{frag}}^A$ when $A$ first sets ${\tt frag\_accepted}^A$ to be true. We claim that $j+g\le 2dm+s-1$, namely, round $j+g$ must be before the end of the execution: When $A$ sets ${\tt frag\_accepted}^A$ to be true, Line~\ref{line_frag_noncomm} must be satisfied, implying that $2d|\sigma^A(x_s)| \geq j-(s-1)+d$. Thus we have $j+g\le j+d \le 2d|\sigma^A(x_s)|+ (s-1)-d+ d \leq 2dm + (s-1)-d+d = 2dm+s-1$. 
		
		By Theorem~\ref{thm_root}, we know that all honest nodes must have the same singleton set ${\tt root\_accepted}$ at Line~\ref{line_final}. Let $r_0$ be the sole element in this set, and let $y$ be the last fragment corresponding to the Merkle root $r_0$. Let $x$ and $z$ be the value of $x^A_s$ and $\sigma^A(y)$, respectively, on $A$ when $A$ first sets ${\tt frag\_accepted}^A$ to be true. Using the exact same arguments as in the proof of Lemma~\ref{lemma:frag1}, one can show that $x = y$. 
		
		\Paragraph{Property for $C$.}
		We next claim the following: Consider any node $B$ on the shortest honest path from $A$ to $C$ (both inclusive). Let $g_1\in[0,d]$ be the honest distance between $A$ and $B$. Then by the end of round $j+g_1$, node $B$ must have forwarded/sent all $s$ fragments corresponding to the Merkle root $r_0$, as well as $z$, to its neighbors. 
		
		The above claim can be proved via an induction, in the same way as in the proof of Lemma~\ref{lemma:frag1} (after replacing ``$g$'' and ``$C$'' with ``$g_1$'' and ``$B$'', respectively). Applying the claim to $C$ shows that by the end of round $j+g$, node $C$ must have sent all $s$ fragments and $z$ to its neighbors. 
		
		\Paragraph{Tracing execution on $C$.}
		Since $z$ is the value of $\sigma^A(y)$ on node $A$ when $A$ first sets ${\tt frag\_accepted}$ to be true, and since $y=x_s^A$ in that round, we have $2d|z|\ge j-(s-1)+d$. With this property, we can now trace all the steps in Algorithm 3 on node $C$ to show that $C$ must set ${\tt frag\_accepted}^D$ to be true in round $j+g$, using the same proof as in Lemma~\ref{lemma:frag1}.
	\end{proof}
}

Building upon Lemma~\ref{lemma:frag1} and \ref{lemma:frag2}, Theorem~\ref{thm_frag} next shows that all honest nodes must agree on whether they accept the fragements:
\begin{theorem}[Agreement on Fragments] \label{thm_frag}
	Consider any execution of Algorithm~\ref{alg_top}, where at least one honest node has $|{\tt root\_accepted}| = 1$ at Line~\ref{line_final}. Then in this execution, all honest nodes must have the same ${\tt frag\_accepted}$ value at Line~\ref{line_final}.
\end{theorem}
\ifthenelse{\boolean{short}}{
}{
	\begin{proof}
		If all honest nodes have ${\tt frag\_accepted} = {\tt false}$ at Line~\ref{line_final}, we are done.
		Otherwise there is some honest node $A$ having ${\tt frag\_accepted} = {\tt true}$. 
		If $A$ is a committee member, then by Lemma~\ref{lemma:frag1}, all honest nodes must have ${\tt frag\_accepted} = {\tt true}$. If $A$ is a non-committee member, then by Lemma~\ref{lemma:frag2} and by Lemma~\ref{lemma:frag1}, all honest nodes must also have ${\tt frag\_accepted} = {\tt true}$ at Line~\ref{line_final}.
	\end{proof}
}

\vspace*{1mm}
\subsection{Proving Theorem~\ref{thm_agreement}}
\label{app:finalproof}
\vspace*{1mm}

\objagree*

\begin{proof}
	We first prove that all honest nodes must return the same object. If all honest nodes output $\bot$, we are done. Otherwise some honest node must satisfy Line~\ref{line_final} with $|\tt root\_accepted| = \{r_0\}$ and $\tt frag\_accpeted = true$, for some $r_0$. Let $A$ be any honest node. Then by Theorem~\ref{thm_root} and Theorem~\ref{thm_frag}, node $A$ must also have ${\tt root\_accepted}^A= \{r_0\}$ and ${\tt frag\_accepted}^A = \tt true$. Next, it suffices to prove that ${\tt all\_frag}^A$ contains all $s$ fragments corresponding to $r_0$. Let $t_0$ be the round during which $A$ first set 
	${\tt frag\_accepted}^A$ to be true. In round $t_0$, since Line~\ref{line_frag_comm} or Line~\ref{line_frag_noncomm} must be satisfied on $A$, we must have $t_{\textnormal{frag}}^A\ne \infty$ and $t_{\textnormal{root}}^A\ne \infty$. This means that $t_{\textnormal{root}}^A$ has already been assigned some value in or before round $t_0$. Since $t_{\textnormal{root}}^A$ is never assigned a value larger than the current round, we must have $t_{\textnormal{root}}^A \leq t_0$. Then by Lemma~\ref{lem_high_push}, in round $t_0$ the most promising push chosen by $A$ at  Line~\ref{line_frag_pick_root} must contain $r_0$. Since $A$ later sets ${\tt frag\_accepted}^A$ to be true in that round, by Line~\ref{line_frag_receive_last}, all the $s$ fragments corresponding to $r_0$ must already be in ${\tt all\_frag}^A$. 
	
	We next prove that if an honest broadcaster $A$ broadcasts an object $O$, then all honest nodes must return $O$. Given we have already proved that all honest nodes must return the same object, it suffices to show that $A$ will return $O$.
	Let the Merkle root for $O$ be $r_0$. By Line~\ref{line_broadcast_sign}, no other Merkle roots will  ever be processed by any honest node, since they do not have a signature from $A$. In round 0, node $A$ must reach and satisfy Line~\ref{line_root_comm}. Then $A$ will add $r_0$ to ${\tt root\_accepted}$ and set $t_{\textnormal{root}}$ to be $0$. Since no other root will ever be processed by $A$, $A$ must have $|{\tt root\_accepted}|=1$ at Line~\ref{line_final}.
	Next, one can trivially follow the steps in Algorithm 3 and verify that during round $s-1$, node $A$ must have $t_{\textnormal{root}}=0$ and $t_{\textnormal{frag}}=s-1$ at Line~\ref{line_frag_comm}. Hence node $A$ must later set ${\tt frag\_accepted}$ to be true.
	Finally, one can trivially verify that at Line~\ref{line_final}, node $A$ must have all $s$ fragments corresponding to $r_0$ in ${\tt all\_frag}$. Putting everything together, $A$ must return $O$.
	
	Finally, it is obvious from the pseudo-code that the algorithm always returns within $2dm+s$ rounds, regardless of whether the committee has any honest member.     
\end{proof}

\end{document}